\documentclass[sigconf,nonacm]{acmart}

\usepackage{tikz}
\usepackage{booktabs}
\usepackage{graphicx}

\setcopyright{none}
\acmConference[Under review]{Under review}{}{}
\acmYear{2026}
\newtheorem{remark}{Remark}

\newcounter{assumption}
\newenvironment{assumption}{\refstepcounter{assumption}\par\medskip\noindent\textbf{A\theassumption.}\enspace}{\par\medskip}
\newenvironment{assumptionstar}[1]{\par\medskip\noindent\textbf{#1.}\enspace}{\par\medskip}

\newcommand{\E}{\mathrm E}
\newcommand{\iidsiminline}{{\scriptstyle\stackrel{\mathrm{iid}}{\sim}}}
\newcommand{\iidsim}{{\stackrel{\mathrm{iid}}{\sim}}}
\newcommand{\ind}{\perp\!\!\!\perp}
\newcommand{\R}{\mathrm{I}\!\mathrm{R}}
\newcommand{\nn}{\nonumber}
\DeclareMathOperator{\var}{var}
\DeclareMathOperator{\cov}{cov}

\begin{document}

\title[EARL: Bipartite Experiments with Partial Assignment]{EARL: Exposure- and Allocation-Reweighted Linear Estimator for Bipartite Experiments with Partial Assignment}

\author{Alexey Kurennoy}
\affiliation{%
  \institution{}
  \city{}
  \country{}
}
\email{}

\begin{abstract}
    Bipartite A/B tests are experiments in which treatment is randomised over one set of units, while outcomes are measured on another. For example, an online marketplace may test a new pricing algorithm on a
    random subset of \emph{items}, while the outcome of interest, say, purchase satisfaction, is measured on \emph{customers}, each of whom interacts with many items.
  
    Existing methods for analysing bipartite experiments assume that every randomisation unit is assigned to treatment or control. In practice, often only a subset participates: platforms cap rollout risk, reserve
    holdout groups, and split their population across concurrent tests. Ignoring the unassigned units biases estimation, while including them requires care.
  
    We construct an unbiased linear estimator for bipartite experiments with partial assignment. Observations must be reweighted not only by the (centred) share of treated connections among participating ones (the
    exposure) but also by the number of participating connections, each inverse-weighted by its participation propensity, so that units whose connections are well covered carry proportionally more weight. The
    resulting estimator, EARL (Exposure- and Allocation-Reweighted Linear), is unbiased, consistent, and asymptotically normal; we devise two asymptotic variance estimators and show that it has minimal variance in
    a natural class of linear estimators. EARL remains unbiased \emph{regardless of the experience the unassigned units receive}, as long as they contribute to expected outcomes additively; in particular, they can be
    allocated to other, non-overlapping tests. Our theory is complemented with a simulation study on two public datasets, in which EARL attains up to six times lower error than the strongest existing baseline and,
    in some configurations, over an order of magnitude lower than Horvitz-Thompson-style alternatives.
\end{abstract}



\maketitle

\section{Introduction}\label{sec:introduction}

Online platforms routinely evaluate changes through randomised experiments (A/B tests). The textbook design randomises the same units on which outcomes are measured. In many important applications, however, the intervention is applied to one population while its effects materialise on another. A marketplace changes the pricing mechanism of \emph{items}, but cares about the satisfaction of \emph{customers} who interact with those items; an advertising platform modifies the auction configuration of \emph{ad campaigns}, but measures the experience of \emph{users} exposed to many campaigns; a video service alters the encoding of \emph{titles}, but tracks the engagement of \emph{viewers}. Experiments of this kind are called \emph{bipartite}: one set of units~--- called \emph{randomisation} or \emph{diversion} units~--- is randomly split between treatment and control, while outcomes are recorded on a different set of \emph{analysis} (or \emph{outcome}) units, and the two sets are connected by a bipartite \emph{experiment graph} \cite{zigler2021bipartite,harshaw2023erl,doudchenko2020causal}. Because each analysis unit is typically connected to many randomisation units, its outcome depends on the whole profile of treatments among its connections; this dependence of each outcome on the treatments of many units invalidates the naive difference-in-means analysis and calls for dedicated estimators.

Existing methods for analysing bipartite experiments almost universally presume that \emph{every} randomisation unit is assigned to either treatment or control. Industrial reality is different: often only a subset of randomisation units participates in any given experiment \cite{kohavi2020trustworthy,tang2010overlapping}. There are at least three reasons. First, platforms limit the risk of changes by enrolling only a fraction (say, 20\%) of units. Second, they run multiple concurrent experiments and sometimes divide the population of randomisation units between them, so that different tests ``own'' different slices of the population. Third, holdout groups are often reserved for long-term measurement. In all these cases, the units left out of the experiment~--- we call them \emph{unassigned} or \emph{unallocated}~--- do not disappear: they remain connected to the analysis units through the experiment graph and keep influencing the measured outcomes. A customer's satisfaction depends on all items they interact with, not only on the items enrolled in our pricing test.

It is tempting to analyse such an experiment by simply deleting the unassigned units from the graph and applying a standard estimator to the reduced graph. This, however, produces a systematically \emph{biased} answer. The estimand of interest is the \emph{global average treatment effect} (GATE)~--- the change in average outcomes if the treatment were rolled out to \emph{all} randomisation units, compared with none. An estimator computed on the reduced graph does not ``know'' that each analysis unit has further, currently unassigned connections that will also switch to the treatment experience after the rollout; as we show in Section~\ref{subsec:estimator}, under the same linear-contributions model that makes EARL unbiased, the exposure-reweighted linear (ERL) estimator of \citet{harshaw2023erl} applied to the reduced graph converges to $q\cdot GATE$, where $q$ is the participation rate, thus having an attenuation of up to 90\% at realistic allocation rates. The problem is a special case of estimation under a misspecified interference graph, which is known to introduce bias that grows with the divergence between the assumed and true graphs \cite{weinstein2026causal}; see \cite{savje2024misspecified} for the general treatment of misspecified exposure mappings. Inverse-probability-weighting (IPW) estimators, on the other hand, can be adapted to remain unbiased, but their weights (and with them the worst-case variance) grow exponentially in the number of connections, and the problem deteriorates further under partial assignment, because the probability of observing an extreme exposure shrinks geometrically with the participation rate.

In this paper, we show that for unbiased and practical estimation of the treatment effect under partial assignment, one has to weight the observations not only by the (centred) exposure of the respective analysis units but also by the inverse-propensity-weighted count of their participating connections, thereby tilting the estimator towards analysis units whose connections are well covered by the experiment while restoring, on average, each unit's full connectivity. We provide the exact recipe for this reweighting, which leads to a novel ERL-type estimator for bipartite experiments with partial assignment that we call the EARL (Exposure- and Allocation-Reweighted Linear) estimator. We prove its unbiasedness, consistency, and asymptotic normality, and devise estimators of its asymptotic variance, which together enable practical inference. We further prove that EARL has the smallest variance in a natural class of linear reweighting estimators, uniformly over the outcome models we consider.

The contributions of this paper are as follows.
\begin{itemize}
    \item \textbf{A novel unbiased estimator for bipartite experiments with partial assignment.} We construct an ERL-type estimator that is unbiased under partial assignment and remains unbiased \emph{whatever experience the unassigned units receive} (e.g., participation in other, non-overlapping experiments), as long as they contribute to the expected outcomes additively It formally reduces to the standard ERL estimator at full assignment, and is provably the minimum-variance member of a natural class of linear estimators. We establish a bias bound without any linearity assumptions, and exact unbiasedness, consistency, and asymptotic normality under a linear-contributions model.
    \item \textbf{Inference machinery.} We devise a closed-form, exactly unbiased estimator of EARL's variance in the spirit of \citet{harshaw2023erl}, and analyse the randomisation-inference variance estimator of \citet{shi2025scalableanalysisbipartiteexperiments} under partial assignment. For the latter, we identify a subtlety specific to partial assignment: the validity requires a null hypothesis stating that expected outcomes are unaffected by \emph{allocation} as well as assignment; we supply a full, self-contained proof under this null.
    \item \textbf{A reproducible semi-synthetic evaluation.} Following the experimental designs of \citet{doudchenko2020causal} and \citet{harshaw2023erl}, we benchmark EARL on a synthetic graph and on two public bipartite graphs (Amazon product reviews and MovieLens), varying the participation rate from 10\% to 100\%. EARL is unbiased in all linear-response scenarios and reduces the estimation error of the strongest existing baseline by up to a factor of six and that of Horvitz--Thompson alternatives, in some configurations, by more than an order of magnitude. We also map out, and explain, the low-allocation regime in which the biased reduced-graph estimator retains a lower total error.
\end{itemize}

The paper is structured as follows. Section~\ref{sec:literature} reviews related work. Section~\ref{subsec:notation} describes our setting and introduces notation. Section~\ref{subsec:estimator} introduces the EARL estimator, illustrates on a worked example why reduced-graph estimation is biased, and states the optimality result. Section~\ref{subsec:analysis} presents our assumptions together with the unbiasedness and consistency theorems. Section~\ref{subsec:inference} develops inference: asymptotic normality and two asymptotic-variance estimators. Section~\ref{sec:simulations} reports the simulation study. Section~\ref{sec:conclusion} concludes. All proofs are collected in the Appendix.

\section{Related Literature}\label{sec:literature}

\textbf{Bipartite experiments.} The bipartite interference framework was formalised by \citet{zigler2021bipartite}, motivated by environmental applications in which interventions on power plants affect health outcomes in surrounding areas; see also \citet{hudgens2008toward,aronow2017estimating,eckles2017design} for the general treatment of interference in experiments. In online platforms, bipartite experiments arise in marketplaces and advertising systems, and a line of work addresses their design: correlation-clustering designs \cite{pouget2019variance}, exposure-optimised designs \cite{harshaw2023erl}, and cluster designs for one-sided randomisation \cite{brennan2022cluster}. \citet{bajari2021multiple} study multiple randomisation designs, which randomise both sides of the graph simultaneously. These design questions are orthogonal to ours; we take the (Bernoulli) design as given and focus on estimation when part of the population is not enrolled.

\textbf{Two families of estimators.} Existing estimators of the GATE in bipartite experiments fall into two broad groups. \emph{IPW (Horvitz--Thompson) estimators} build on exposure mappings \cite{aronow2017estimating}: each analysis unit's treatment configuration is summarised by an exposure level, and outcomes are reweighted by the inverse probability of the observed exposure. \citet{zigler2021bipartite} and \citet{doudchenko2020causal} develop this approach for bipartite structures, the latter using generalised propensity scores \cite{hirano2004propensity}. The appeal of IPW is minimal outcome-model assumptions; its weakness is variance: the probability of observing a ``pure'' (all-treated or all-control) exposure decays geometrically in the number of connections, so the weights explode. \emph{ERL-type estimators} instead posit a response structure that is linear in the exposure. \citet{harshaw2023erl} introduce the exposure-reweighted linear (ERL) estimator, prove its unbiasedness, consistency and asymptotic normality under a linear exposure-response model, and construct an unbiased variance estimator. 
\citet{lu2025design} give a design-based analysis with a conservative variance estimator and covariate adjustment. 
All of the above assume full assignment of the randomisation units.

\textbf{Partial participation.} Closest to our work, \citet{tan2025estimating} study bipartite experiments with \emph{partial eligibility}: a fixed, known subset of treatment-side units is eligible for randomisation, while the remaining units are never treated. They define total-effect estimands restricted to the eligible side (PTTE) and to spillovers on ineligible units (STTE), and estimate them with an ensemble of outcome models, generalised propensity scores, and machine-learning components, with bootstrap inference. Our setting differs in three substantive ways. First, in our framework participation itself is randomised (with probability $q$), which is common when an experimentation platform carves the population into slices for concurrent tests \cite{tang2010overlapping}; eligibility in \citet{tan2025estimating} is a deterministic attribute. Second, and more importantly, we make no assumption about what the unassigned units experience, for example, they may be exposed to \emph{other} interventions, as happens when they are allocated to concurrent experiments, whereas ineligible units in \citet{tan2025estimating} are never directly assigned treatment, although all units continue interacting and may be affected through spillovers. Third, our target is the classical GATE for a rollout to the \emph{whole} population, and our machinery is design-based with closed-form weights, exact unbiasedness, central limit theorems, and analytic variance estimators, rather than model-based ensembles with bootstrap uncertainty. The two approaches are thus complementary: theirs favours flexible outcome modelling under fixed eligibility, ours favours transparent, assumption-light inference under randomised partial assignment.

\textbf{Graph misspecification.} Running a full-assignment estimator on the reduced graph is equivalent to analysing the experiment under a misspecified interference graph. \citet{weinstein2026causal} bound the bias arising from misspecified interference networks and show it grows with the divergence between the assumed and true networks; \citet{savje2024misspecified} studies when misspecified exposure mappings still permit meaningful estimation. Our contribution can be seen as replacing an (incorrectly) reduced graph with a weighting scheme on the full graph that restores unbiasedness at a modest variance cost.

\section{Proposed Estimator}\label{sec:estimator}

\subsection{Problem Setup and Notation}\label{subsec:notation}

Let $\mathcal A$ and $\mathcal R$ be the populations of $N$ analysis and $M$ randomisation units, respectively, and let $\mathcal C\subset \mathcal A \times \mathcal R$ be the set of connections between them. Together, the triple $\langle \mathcal A,\,\mathcal R,\,\mathcal C\rangle$ forms a bipartite graph that is commonly referred to as the \emph{experiment graph}. In the literature, analysis units are also called \emph{outcome} or \emph{measurement} units, and randomisation units are also called \emph{diversion}, \emph{treatment}, or \emph{intervention} units \cite{zigler2021bipartite,harshaw2023erl,doudchenko2020causal,tan2025estimating}. In our running marketplace example, the randomisation units are items, the analysis units are customers, and a connection $(a,\,r)$ records that customer $a$ interacts with item $r$. We use $d_\mathcal{A}$ to denote the maximum number of randomisation units connected to a given analysis unit and $d_\mathcal{R}$~~---  the maximum number of analysis units connected to a given randomisation unit.

Randomisation units are allocated for participation in the experiment independently at random with probability $q\in(0,\,1)$. Let $\mathbf{S} =(S_1,\,\ldots,\,S_M)$ be the vector of allocation indicators: $S_r$ is 1 if the $r$-th randomisation unit participates in the experiment and $0$ otherwise, $S_r\iidsiminline \mathrm{Bernoulli}(q)$. Randomisation units participating in the experiment are randomly and independently assigned to the treatment group with probability $p\in(0,\,1)$. Let $\mathbf Z=(Z_1,\,\ldots,\,Z_M)$ be the vector of assignment indicators: $Z_r$ is 1 if the $r$-th unit is (or would be, should it participate) assigned to the treatment, $Z_r\iidsiminline \mathrm{Bernoulli}(p)$. In industrial experimentation systems both indicators are computed by deterministic hashing of unit identifiers \cite{tang2010overlapping,kohavi2020trustworthy}, so $Z_r$ is well defined (though immaterial) even for units that do not participate. We assume that allocation and assignment are performed independently of each other, i.e., $\mathbf S \ind \mathbf{Z}$. While existing methods for analysing bipartite experiments assume that all randomisation units are assigned (i.e., $q=1$), in this paper we consider the partial assignment case ($q < 1$). It means that we have three groups of randomisation units: unassigned ($S_r = 0$), treatment ($S_r=1$ and $Z_r=1$), and control ($S_r=1$ and $Z_r=0$):
\begin{center}
    \begin{tabular}{cccc}
    \hline
    Group & \multicolumn{2}{c}{Indicator Values} & Expected Count\\
    & $S_r$ & $Z_r$ & \\
    \hline
    Unassigned & 0 & any & $(1-q) M$\\
    Treatment & 1 & 1 & $pq M$\\
    Control & 1 & 0 & $(1-p)q M$\\
    \hline
\end{tabular}
\end{center}
Note that even though unassigned units do not participate in the experiment, they are still connected to analysis units via the experiment graph and influence the outcomes associated with analysis units. In the marketplace example: a customer's satisfaction is shaped by all the items they interact with, including items that our pricing test never touches.

For a given analysis unit $a\in\mathcal A$, let $Y_a$ be the outcome (random) variable (which we assume to be integrable) and let $e_a(\mathbf{s},\,\mathbf{z})$
be the expected outcome conditional on $\mathbf S = \mathbf s$ and $\mathbf Z = \mathbf z$, i.e.,
\begin{equation}\label{eq:ea}
    e_a(\mathbf{s},\,\mathbf{z}) = \E[Y_a\mid \mathbf{S} = \mathbf{s},\,\mathbf{Z}=\mathbf{z}],\quad \mathbf{s}\in \{0,\,1\}^M,\, \mathbf{z}\in \{0,\,1\}^M.
\end{equation}
Furthermore, for each $a\in\mathcal{A}$, define the function $\mu_a$ as follows,
\begin{eqnarray}\label{eq:ma}
    \mu_a(x,\,y) &=& \E_{S_r\iidsim \mathrm{Bern}(x),\,Z_r\iidsim \mathrm{Bern}(y),\,\mathbf S\ind \mathbf Z}[e_a(\mathbf S,\,\mathbf Z)]\nonumber\\
    &=& \sum_{\textbf{s},\textbf{z}\in \{0,\,1\}^M}\prod_{r=1}^Mx^{s_r}(1-x)^{1-s_r}y^{z_r}(1-y)^{1-z_r}e_a(\textbf{s},\,\textbf{z}),
\end{eqnarray}
$x\in[0,\,1]$, $y\in[0,\,1]$, where we adopt the convention $0^0=1$ (so that, e.g., the expression is well defined at $x\in\{0,\,1\}$). The function $\mu_a$ is the expected outcome of unit $a$ in a hypothetical experiment with allocation rate $x$ and treatment rate $y$; it is a polynomial in $(x,\,y)$ and thus infinitely differentiable. Finally, let the function $\mu$ be the average of $\mu_a$ over the population of analysis units $\mathcal{A}$, i.\,e.,
\begin{equation}\label{eq:mu}
    \mu(x,\,y) = \frac 1 N \sum_{a\in\mathcal{A}} \mu_a(x,\,y),\qquad x\in[0,\,1],\, y\in[0,\,1].
\end{equation}

The goal of the experimenter is to estimate the \emph{global average treatment effect},
\begin{equation}\label{eq:gate}
    GATE = \frac 1 N\sum_{a\in\mathcal A} \left(\mu_a(1,\,1) - \mu_a(1,\,0)\right) = \mu(1,\,1) - \mu(1,\,0),
\end{equation}
i.e., the contrast between average expected outcomes when \emph{all} randomisation units receive treatment and when all of them receive control. This is the quantity relevant for the launch decision: it describes the world after a full rollout.

In what follows, $\mathcal D(a)$ denotes the set of randomisation units that are connected to the analysis unit $a$, i.e.,
\begin{equation*}
    \mathcal{D}(a) = \{r\in\mathcal{R}\colon (a,\,r)\in \mathcal{C}\}.
\end{equation*}
For example, if customer $a$ interacts with five items, of which two are allocated to the experiment and one of those two is treated, then $|\mathcal D(a)|=5$ and the realised allocation/assignment pattern on $\mathcal D(a)$ has two participating units, one of which is treated.

All random variables are assumed to be defined on a common probability space $(\Omega,\,\mathcal F,\,\Pr)$. In all asymptotic statements, the design parameters $p$ and $q$ are held fixed as the number of units grows.

\subsection{Exposure- and Allocation-Reweighted Linear Estimator}\label{subsec:estimator}

Let $H_a$ be the so-called exposure of analysis unit $a$, i.e., the share of treated units among the randomisation units $a$ is connected to,
\begin{equation}\label{eq:Ha}
    H_a = \frac{1}{|\mathcal{D}(a)|}\sum_{r\in\mathcal{D}(a)} Z_r.
\end{equation}
When all randomisation units are assigned ($q=1$), the standard ERL estimator \cite{harshaw2023erl} is defined as the average of outcomes $Y_a$ weighted by the standardised exposures,
\begin{equation}{\label{eq:erl}}
    \hat\tau_{ERL} = \frac 1 N \sum_{a\in\mathcal A} \frac{H_a - \E[H_a]}{\var \left(H_a\right)}Y_a =\frac{1}{N}\sum_{a\in\mathcal A}\sum_{r\in\mathcal{D}(a)}\frac{Z_r - p}{p(1-p)}Y_a.
\end{equation}
The interpretation of this definition is that units with a greater share of treated randomisation units among their connections contribute more to the effect estimate. As demonstrated in \cite{harshaw2023erl}, ERL is unbiased for GATE under a linear exposure-response assumption.

Under partial assignment ($q<1$) two additional quantities describe how well the experiment covers the connections of $a$. Let $G_a$ be the share of participating units among the connections of $a$ (the \emph{allocation coverage}),
\begin{equation}\label{eq:Ga}
    G_a = \frac{1}{|\mathcal{D}(a)|}\sum_{r\in\mathcal{D}(a)} S_r,
\end{equation}
and let $\tilde H_a$ be the exposure computed \emph{within} the participating connections,
\begin{equation}\label{eq:Hta}
    \tilde H_a = \frac{\sum_{r\in\mathcal{D}(a)} S_rZ_r}{\sum_{r\in\mathcal{D}(a)} S_r}\quad\textup{if } G_a>0.
\end{equation}
For isolated units ($\mathcal D(a)=\emptyset$) we set $H_a = G_a = 0$.
The estimator we propose weights each observation by the product of the allocation coverage and the centred within-experiment exposure, scaled by the inverse allocation rate:
\begin{eqnarray}\label{eq:earl}
    \hat \tau_{EARL} &=& \frac 1 N \sum_{a\in \mathcal A} |\mathcal D(a)|\,\frac{G_a\,(\tilde H_a - p)}{q\,p(1-p)}\, Y_a \nonumber\\
    &=& \frac 1 N \sum_{a\in \mathcal A} \sum_{r\in\mathcal{D}(a)}\frac{S_r}{q}\cdot\frac{Z_r - p}{p(1-p)}\, Y_a,
\end{eqnarray}
with the convention that the first form is $0$ whenever $G_a=0$ (as the second form makes explicit). We call it the Exposure- and Allocation-Reweighted Linear (EARL) estimator. In comparison with the standard ERL estimator \eqref{eq:erl}, EARL (i) computes exposures only over the connections actually enrolled in the experiment, and (ii) additionally tilts the weights towards those analysis units that have a high share of participating randomisation units among their connections, upweighting every participating connection by the inverse of its allocation propensity. Setting $q=1$ in \eqref{eq:earl} formally recovers ERL; throughout the paper, however, we treat the partial-assignment case $q\in(0,\,1)$, as full assignment is covered by the standard ERL theory \cite{harshaw2023erl}.

Three practical remarks are in order. First, EARL only requires the assignment indicators $Z_r$ of \emph{participating} units ($S_r=1$), which are always observed; the ``would-be'' assignments of unassigned units are not needed. Second, the estimator is linear in the outcomes and its weights depend only on design quantities ($q$, $p$) and the experiment graph, so it can be computed in a single pass over the graph edges~--- an $O(|\mathcal C|)$ computation that fits distributed query engines, similarly to \cite{shi2025scalableanalysisbipartiteexperiments}. Third, as will be demonstrated in Section~\ref{subsec:analysis}, EARL is unbiased (and consistent) for GATE \eqref{eq:gate} \emph{regardless of the actual experience the unassigned units receive}, provided that experience enters the expected outcomes additively (Assumptions~A\ref{a:mu} and A\ref{a:4} below). They can be, for example, exposed to the treatment (as in the case of a ``backtest'') or allocated to other, disjoint experiments. This has important practical consequences: the experimenter need not ensure that unassigned units keep receiving the control experience, and thus retains much greater experimentation throughput.

\medskip
\noindent \textbf{Example.}
Consider an experiment in which unassigned randomisation units are allocated to a concurrent, non-overlapping test. A randomisation unit adds $0$, $1$, or $\alpha$ to the outcome of each connected analysis unit depending on whether the randomisation unit belongs to control, treatment, or the group of unassigned units, respectively:
\begin{equation*}
    Y_a = \sum_{r\in\mathcal D(a)}\bigl(\alpha(1-S_r) + S_rZ_r\bigr).
\end{equation*}
The effect $\alpha$ of unassigned units does not have to be 0, for example, because of their participation in a concurrent A/B test. The GATE \eqref{eq:gate} in this case equals $N^{-1}\sum_{a\in\mathcal{A}}|\mathcal{D}(a)|$.
One natural way to apply the existing full-assignment machinery to GATE estimation under partial assignment is to ignore the presence of unassigned units, i.e., to exclude them and compute the standard ERL estimator \cite{harshaw2023erl} on the reduced graph,
\begin{equation}\label{eq:erl-drop}
    \hat \tau_{drop} = \frac 1 N \sum_{a\in\mathcal{A}}\sum_{r\in\mathcal{D}(a)}S_r\frac{Z_r-p}{p(1-p)}Y_a.
\end{equation}
It is easy to see that this approach leads to a biased estimator in this example. Using the independence of the indicators and $\E[Z_r-p]=0$, all cross terms ($r\neq r'$) and all $\alpha$-terms vanish, and
\begin{eqnarray*}
    \E[\hat\tau_{drop}] &=& \frac {1}{Np(1-p)}\sum_{a\in\mathcal{A}} \sum_{r\in\mathcal{D}(a)}\E\left[S_r(Z_r-p)S_rZ_r  \right]\\
    &=& \frac 1 {Np(1-p)}\sum_{a\in\mathcal{A}}\left|\mathcal{D}(a)\right|q\,p(1-p)\\
    &=&q \cdot GATE.
\end{eqnarray*}
The reduction of the graph leads to an underestimation of the true treatment effect: the estimator does not ``know'' that the actual number of connections is larger and that the analysis units will be affected to a greater extent when the treatment is rolled out to \emph{all} randomisation units of the original graph. In contrast, the EARL estimator \eqref{eq:earl} accounts for the allocation propensity and has no bias:
\begin{eqnarray*}
    \E[\hat\tau_{EARL}] &=& \frac {1}{Nqp(1-p)}\sum_{a\in\mathcal{A}} \sum_{r\in\mathcal{D}(a)}\E\left[S_r(Z_r-p)S_rZ_r  \right]\\
    &=& \frac 1 {Nqp(1-p)}\sum_{a\in\mathcal{A}}\left|\mathcal{D}(a)\right|q\,p(1-p)
    = GATE.
\end{eqnarray*}
Note that the computed expectations of $\hat \tau_{drop}$ and $\hat\tau_{EARL}$ above do not depend on $\alpha$. It means that even if the unassigned units receive the same experience as control ($\alpha=0$), ignoring them still leads to biased estimation in this example; and on the other hand, our proposed estimator remains unbiased regardless of the specific value of $\alpha$.

\medskip
The example suggests that the bias of the reduced-graph estimator could be repaired by dividing by $q$~--- and indeed, EARL \eqref{eq:earl} equals $\hat\tau_{drop}/q$. This simple correction is, however, not ad hoc: the next proposition shows that among \emph{all} linear reweighting estimators that correct the allocation in an unbiased way, EARL is the most efficient, uniformly over the class of outcome models introduced in Section~\ref{subsec:analysis}. Consider the family
\begin{equation}\label{eq:class}
    \hat\tau_{c} = \frac 1 N \sum_{a\in \mathcal A} \sum_{r\in\mathcal{D}(a)}c(S_r)\,\frac{Z_r - p}{p(1-p)}\, Y_a,\qquad c\colon\{0,\,1\}\mapsto\R,
\end{equation}
which contains the reduced-graph ERL ($c(s)=s$), EARL ($c(s)=s/q$), and, e.g., the estimator with the standardised-allocation weight $c(s)=(s-q)/(q(1-q))$.

\begin{proposition}\label{prop:optimality}
    Let the experiment graph contain at least one connection ($\mathcal C\neq\emptyset$), and let Assumptions A\ref{a:1}, A\ref{a:mu}, and A\ref{a:no-autocorrelation} of Section~\ref{subsec:analysis} hold (the latter includes square integrability of the outcomes). Suppose, in addition, that the conditional second moments $m_a(\mathbf s,\,\mathbf z) := \E[Y_a^2\mid \mathbf S=\mathbf s,\,\mathbf Z=\mathbf z]$ exclude would-be assignments in the sense of Assumption~A\ref{a:mu}(ii): for every $a\in\mathcal A$, $m_a(\mathbf s,\,\mathbf z) = m_a(\mathbf s,\,\mathbf z')$ whenever $z_r=z_r'$ for all $r$ with $s_r=1$. Then:
    \begin{itemize}
        \item[\textup{(a)}] $\hat\tau_c$ satisfies $\E[\hat\tau_c]=\E[\hat\tau_{EARL}]$ for every outcome model satisfying the above if and only if $q\,c(1)=1$; in particular, under the conditions of Theorem~\ref{th:unbiasedness}(b) below, $\hat\tau_c$ is unbiased for GATE if and only if $q\,c(1)=1$;
        \item[\textup{(b)}] for every $c$ with $q\,c(1)=1$,
        \begin{equation*}
        \var(\hat\tau_c) = \var(\hat\tau_{EARL}) + c(0)^2\,\var(\hat\tau_U),
        \end{equation*}
        where $\hat\tau_U = N^{-1}\sum_{a}\sum_{r\in\mathcal D(a)}(1-S_r)\frac{Z_r-p}{p(1-p)}Y_a$. Consequently, $\var(\hat\tau_c) \ge \var(\hat\tau_{EARL})$, with equality if and only if $c(0)=0$ or $\var(\hat\tau_U)=0$.
    \end{itemize}
\end{proposition}

Like Assumption~A\ref{a:mu}(ii) itself, the second-moment condition holds under the standard randomisation construction, in which the assignment indicators are exogenous to all outcome disturbances and an unassigned unit's would-be assignment is never materialised. The proof (Appendix~\ref{app:optimality}) rests on an exact orthogonal decomposition: any admissible $\hat\tau_c$ equals EARL plus a mean-zero term that is uncorrelated with it, so deviating from $c(0)=0$ can only add noise. The guarantee is confined to the class \eqref{eq:class}.

\subsection{Unbiasedness and Consistency}\label{subsec:analysis}

We begin by stating and discussing the assumptions we use to prove the unbiasedness of the EARL estimator \eqref{eq:earl}. The first assumption formalises the experimental design described in Section~\ref{subsec:notation}.

\begin{assumption}\label{a:1}
\begin{itemize}
\item The allocation of randomisation units to the experiment is carried out independently at random with probability $q\in(0,\,1)$:
\[
    S_1,\ldots,S_M \quad \iidsim \quad \mathrm{Bernoulli}(q).
\]
\item
The assignment of randomisation units to the treatment group is done independently at random with probability $p\in(0,\,1)$:
\[
    Z_1,\ldots,Z_M \quad \iidsim \quad \mathrm{Bernoulli}(p).
\]
\item
Allocation to the experiment and assignment to treatment are conducted independently, i.e., the indicator vectors
$\mathbf S=(S_1,\ldots,S_M)$ and $\mathbf Z=(Z_1,\allowbreak\ldots,\allowbreak Z_M)$ are independent:
\[
    \mathbf S \ind \mathbf Z.
\]
\end{itemize}
\end{assumption}

Somewhat differently from \cite{harshaw2023erl}, we will not maintain a linear response assumption in the strict sense. In our theory, outcomes $Y_a$, $a\in\mathcal A$, are not required to be linear functions of the exposure $\tilde H_a$ and/or the allocation coverage $G_a$; in fact, they do not even have to be measurable with respect to them. Instead, we impose a set of (jointly) weaker conditions on the \emph{conditional expectations} of the outcomes. The first of them states that the expected outcome of an analysis unit depends only on the allocation and assignment of the randomisation units it is directly connected to, and that would-be assignments of unassigned units are immaterial.

\begin{assumption}\label{a:mu}
For every analysis unit $a\in\mathcal A$, the conditional expected outcome $e_a$ defined in \eqref{eq:ea} satisfies:
\begin{itemize}
    \item[\textup{(i)}] \textup{(locality)} $e_a(\mathbf s,\,\mathbf z) = e_a(\mathbf s',\,\mathbf z')$ whenever $s_r=s_r'$ and $z_r=z_r'$ for all $r\in\mathcal D(a)$;
    \item[\textup{(ii)}] \textup{(exclusion of would-be assignments)} $e_a(\mathbf s,\,\mathbf z) = e_a(\mathbf s,\,\mathbf z')$ whenever $z_r=z_r'$ for all $r\in\mathcal D(a)$ with $s_r=1$.
\end{itemize}
\end{assumption}

Part (i) rules out ``higher-order'' effects on the expected outcome that travel through the graph beyond direct connections. Figure~\ref{fig:m-structure} shows the canonical ``M-structure'' in which such an effect could arise: randomisation unit $r_3$ is not connected to analysis unit $a_1$, yet it could influence $Y_{a_1}$ indirectly~--- treatment of $r_3$ changes the behaviour of $a_3$, which alters the state of the shared unit $r_2$ (e.g., its price or inventory), which in turn feeds back into $a_1$'s outcome. Assumption~A\ref{a:mu}(i) asserts that such feedback loops are absent, or negligible in practice, \emph{on average}: it constrains conditional means only, not the outcomes themselves. This is a bipartite analogue of the neighbourhood-interference assumptions standard in the literature \cite{aronow2017estimating,harshaw2023erl,shi2025scalableanalysisbipartiteexperiments}, and it is weaker than requiring $Y_a$ to be a function of the treatments in $\mathcal D(a)$. Part (ii) merely says that the treatment bucket a unit \emph{would have} received, had it participated, cannot influence outcomes when the unit does not participate. Since would-be assignments are never materialised (in a hash-based system they are never even computed for non-participants), (ii) holds under the standard randomisation construction in which the assignment indicators are exogenous to all outcome disturbances; we state it explicitly because the would-be indicators $Z_r$ are formally defined for all units, and Assumption~A\ref{a:1} alone does not rule out a dependence of the outcome distribution on them.

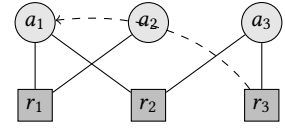
\begin{figure}
\centering
\begin{tikzpicture}[xscale=1.5, yscale=1.1,
  aunit/.style={circle, draw, fill=black!10, inner sep=2pt, minimum size=14pt},
  runit/.style={rectangle, draw, fill=black!25, inner sep=3pt, minimum size=12pt}]
  \node[aunit] (a1) at (0,1) {$a_1$};
  \node[aunit] (a2) at (1,1) {$a_2$};
  \node[aunit] (a3) at (2,1) {$a_3$};
  \node[runit] (r1) at (0,0) {$r_1$};
  \node[runit] (r2) at (1,0) {$r_2$};
  \node[runit] (r3) at (2,0) {$r_3$};
  \draw (a1) -- (r1);
  \draw (a1) -- (r2);
  \draw (a3) -- (r2);
  \draw (a3) -- (r3);
  \draw (a2) -- (r1);
  \draw[dashed,->,bend right=35] (r3) to (a1);
\end{tikzpicture}
\caption{An ``M-structure'' in the experiment graph: circles are analysis units, squares are randomisation units, solid lines are connections. A higher-order effect of $r_3$ on $Y_{a_1}$ (dashed) via $a_3$ and the shared unit $r_2$ is ruled out (on average) by Assumption~A\ref{a:mu}(i).}
\label{fig:m-structure}
\end{figure}

Under Assumption~A\ref{a:mu} the function $\mu_a$ defined in \eqref{eq:ma} admits the decomposition
\begin{eqnarray}\label{eq:mu-decomp}
    \mu_a(x,\,y) = \frac 1 {2^{2(M-|\mathcal D(a)|)}}\sum_{\textbf{s},\,\textbf{z}\in\{0,\,1\}^M} \Bigg( \prod_{r\in \mathcal{D}(a)} x^{s_r}(1-x)^{1-s_r} \nn\\{}\times \prod_{r\in \mathcal{D}(a)}
    y^{z_r}(1-y)^{1-z_r}e_a(\textbf{s},\,\textbf{z})\Bigg)
\end{eqnarray}
for all $x,\,y\in[0,\,1]$ (where, by convention, a product over an empty set equals 1): the coordinates outside $\mathcal D(a)$ integrate out.

The last assumption before we present our unbiasedness result specifies a linear structure for the conditional expected outcomes as functions of the allocation and assignment indicators.

\begin{assumption}\label{a:4}
For each analysis unit $a\in\mathcal{A}$, the conditional expected outcome $e_a$ defined in \eqref{eq:ea} is additive in the contributions of individual randomisation units: there exist coefficients $\beta^a_{r,\,j}\in\R$, $r=1,\,\ldots,\,M$, $j=0,\,1,\,2,\,3$, such that
\begin{equation*}
    e_a(\textbf{s},\,\textbf{z}) = \sum_{r=1}^M\left(\beta^a_{r,\,0}+\beta^a_{r,\,1}(1-s_r)+\beta^a_{r,\,2}s_rz_r+\beta^a_{r,\,3}s_r(1-z_r)\right)
\end{equation*}
for all $\textbf{s}\in\{0,\,1\}^M$ and $\textbf{z}\in\{0,\,1\}^M$.
\end{assumption}
Assumption~A\ref{a:4} says that the expected outcome of unit $a$ is a sum of contributions of randomisation units, where each unit contributes one of three values depending on the group it falls into: unassigned ($\beta^a_{r,\,0}+\beta^a_{r,\,1}$), treatment ($\beta^a_{r,\,0}+\beta^a_{r,\,2}$), or control ($\beta^a_{r,\,0}+\beta^a_{r,\,3}$). Crucially, unassigned ($s_r=0$) and control ($s_r=1$, $z_r=0$) units are allowed to have \emph{different} contributions. In other words, we do not assume that unassigned units deliver the control experience: they are permitted to be allocated to other, non-overlapping experiments (i.e., experiments randomising disjoint sets of randomisation units). Note also that the coefficients may vary arbitrarily across both $a$ and $r$, allowing rich effect heterogeneity, and that A\ref{a:4} restricts conditional means only. (Under A\ref{a:mu}, the coefficients $\beta^a_{r,\,j}$, $j\ge 1$, of units $r\notin\mathcal D(a)$ can be taken to be zero.)

We are now in a position to show that the EARL estimator is unbiased.
\begin{theorem}\label{th:unbiasedness}
    Let the random variables $S_r,\,Z_r\colon\Omega\mapsto\{0,\,1\}$, $r=1,\,\ldots,\,M$, satisfy Assumption A\ref{a:1}. Suppose the random variables $Y_a$, $a\in\mathcal A$, are integrable and Assumption~A\ref{a:mu} is true. Let the function $\mu\colon[0,\,1]\times[0,\,1]\mapsto\R$ be defined by relations \eqref{eq:ea}--\eqref{eq:mu}.

    Then
    \begin{itemize}
    \item[\textup{(a)}] the estimator $\hat\tau_{EARL}$ defined by \eqref{eq:earl} and the global average treatment effect (GATE) defined by \eqref{eq:gate} satisfy the following bound,
    \begin{eqnarray}\label{eq:bound}
        \left|\E[\hat\tau^{\phantom{1}}_{EARL}] - GATE\right| \le \sup_{x,\,y\in(0,\,1)}\Bigg(\left|\frac{\partial^3\mu}{\partial x^2\partial y}(x,\,y)\right| \nn\\ {}+ \left|\frac{\partial^3\mu}{\partial  x\partial y^2}(x,\,y)\right|\Bigg),
    \end{eqnarray}
    \item[\textup{(b)}] if, in addition, Assumption A\ref{a:4} holds true, $\hat \tau_{EARL}$ is unbiased,
    \begin{equation}\label{eq:unbiasedness}
        \E[\hat\tau^{\phantom{1}}_{EARL}] = GATE.
    \end{equation}
    \end{itemize}
\end{theorem}

Part (a) quantifies the bias without any linearity assumption: the estimator's expectation and the GATE are both values of the mixed derivative $\partial^2\mu/\partial x\partial y$ at (possibly different) interior points, so the bias is controlled by the third derivatives of $\mu$~--- a measure of non-linearity of the response surface. Under the additive-contributions model these third derivatives vanish identically, giving part (b).

The consistency result proven below (Theorem~\ref{th:consistency}) uses an extra assumption which states that, once we condition on the allocation and assignment of the randomisation units, no systematic co-movement remains between the outcomes of different analysis units: all cross-unit dependence is channelled through $(\mathbf S,\,\mathbf Z)$.

\begin{assumption}\label{a:no-autocorrelation} The outcome variables $Y_a$, $a\in\mathcal A$, are square integrable (which makes the conditional moments below well defined) and, conditional on $(\mathbf S,\,\mathbf Z)$, the errors
\begin{equation}\label{def:epsa}
    \varepsilon_a = Y_a - e_a(\mathbf{S},\,\mathbf{Z}),\qquad a\in\mathcal{A},
\end{equation}
are uncorrelated:
\begin{equation*}
    \E[\varepsilon_a \varepsilon_{a'}\mid \mathbf{S},\,\mathbf{Z}] = \E[\varepsilon_a\mid \mathbf{S},\,\mathbf{Z}] \cdot \E[\varepsilon_{a'}\mid \mathbf{S},\,\mathbf{Z}]\quad \forall\,a\neq a'\in\mathcal{A}.
\end{equation*}
\end{assumption}

The next theorem demonstrates that EARL is consistent as long as the experiment graph is not too dense.
\begin{theorem}\label{th:consistency}
    Let the random variables $S_r,\,Z_r\colon\Omega\mapsto\{0,\,1\}$, $r=1,\,\ldots,\,M$, satisfy Assumption A\ref{a:1}. Suppose there exists $L>0$ such that $\E[Y_a^2] \le L$ for all $a\in\mathcal A$ (uniformly across all values of $N$). Let Assumptions~A\ref{a:mu}, A\ref{a:4}, and A\ref{a:no-autocorrelation} hold true.

    Then, if $d^{\phantom{1}}_{\mathcal{R}}d^3_{\mathcal{A}}=o(N)$, the EARL estimator defined in \eqref{eq:earl} is consistent for GATE \eqref{eq:gate},
    \begin{equation*}
        \hat\tau^{\phantom{1}}_{EARL} \stackrel{p}{\to} GATE\qquad \textup{as }N\to\infty.
    \end{equation*}
\end{theorem}

The sparsity requirement $d_{\mathcal R}^{\phantom{1}}d_{\mathcal A}^3=o(N)$ coincides with the corresponding condition for the full-assignment ERL estimator \cite[Theorem~4.2]{harshaw2023erl}: partial assignment affects only the constants in the mean-squared-error bound (through the factor $q^{-1}$ in the weights), not the admissible graph-density regime.

\subsection{Inference}\label{subsec:inference}
In this section, we prove the asymptotic normality of the EARL estimator and devise estimators for its asymptotic variance. This enables statistical testing and confidence intervals based on the EARL estimator.

\subsubsection{Asymptotic Normality}

To prove asymptotic normality we strengthen the decorrelation requirement of Assumption A\ref{a:no-autocorrelation}, asking that the errors $\varepsilon_a$, $a\in\mathcal A$, defined in \eqref{def:epsa} be independent across observations and from the allocation and assignment indicators.

\begin{assumptionstar}{A\ref{a:no-autocorrelation}$^*$}
The errors $\{\varepsilon_a\}_{a\in\mathcal A}$, allocation indicators $\mathbf S$, and assignment indicators $\mathbf Z$ are jointly independent.
\end{assumptionstar}

Under A\ref{a:no-autocorrelation}$^*$, the products $\varepsilon_a\varepsilon_{a'}$ are automatically integrable ($\E[|\varepsilon_a\varepsilon_{a'}|] = \E[|\varepsilon_a|]\,\E[|\varepsilon_{a'}|]$), so the conditional uncorrelatedness required by Assumption~A\ref{a:no-autocorrelation} holds for merely integrable outcomes; since every result below that invokes A\ref{a:no-autocorrelation}$^*$ also bounds the fourth moments of the outcomes, Assumption~A\ref{a:no-autocorrelation} holds in full in all those settings.

The proof of asymptotic normality also relies on the following (and last) assumption of this subsection. It rules out practically irrelevant edge cases in which the GATE can be estimated at a faster than parametric rate. See the discussion in \cite[p.~475]{harshaw2023erl} for specific scenarios that are excluded by this condition.

\begin{assumption}\label{a:non-degeneracy} \textit{(Asymptotic Non-Degeneracy)} The normalised variance of the EARL estimator is bounded away from zero, i.e., there exists $\delta >0$ such that $N\cdot\var(\hat\tau_{EARL}) > \delta$ for all sufficiently large $N$.
\end{assumption}

\begin{theorem}\label{th:asymptotic-normality}
    Let the random variables $S_r,\,Z_r\colon\Omega\mapsto\{0,\,1\}$, $r=1,\,\ldots,\,M$, satisfy Assumption A\ref{a:1}. Suppose there exists $L>0$ such that $\E[Y_a^4] \le L$ for all $a\in\mathcal A$ (uniformly across all values of $N$) and Assumptions~A\ref{a:mu}, A\ref{a:4}, A\ref{a:no-autocorrelation}$^*$, and A\ref{a:non-degeneracy} are true.

    Then, if $d^{4}_{\mathcal{R}}d^{10}_{\mathcal{A}}=o(N)$, the EARL estimator defined in \eqref{eq:earl} is asymptotically normal,
    \begin{equation*}
        \frac{\hat\tau^{\phantom{1}}_{EARL} - GATE}{\sqrt{\var\left(\hat\tau^{\phantom{1}}_{EARL}\right)}}\stackrel{d}{\to} \mathcal{N}(0,\,1)\qquad \textup{as }N\to\infty.
    \end{equation*}
\end{theorem}

\subsubsection{Asymptotic Variance Estimation}\label{subsubsec:variance}
An estimator of the asymptotic variance of $\hat\tau^{\phantom{1}}_{EARL}$ can be constructed similarly to \cite[Section 5.1]{harshaw2023erl}. For that purpose, we will need a stronger version of Assumption~A\ref{a:4} under which the contributions of the randomisation units connected to a given analysis unit are homogeneous. Define the \emph{joint coverage--exposure} of unit $a$,
\begin{equation}\label{eq:Fa}
    F_a = \frac{1}{|\mathcal{D}(a)|}\sum_{r\in\mathcal{D}(a)} S_rZ_r = G_a\tilde H_a,
\end{equation}
the share of $a$'s connections that participate \emph{and} are treated (with $F_a=G_a=0$ if $\mathcal D(a)=\emptyset$).

\begin{assumptionstar}{A\ref{a:4}$^*$}
For every $a\in\mathcal A$ there exist $B^a_j\in\R$, $j=0,\,1,\,2$, such that
\begin{equation*}
    e_a(\mathbf s,\,\mathbf z) = B^a_0 + B^a_1 \, G_a + B^a_2 \, F_a ,
\end{equation*}
where $G_a$ and $F_a$ are evaluated at $(\mathbf s,\,\mathbf z)$.
\end{assumptionstar}

Assumption~A\ref{a:4}$^*$ is exactly the special case of Assumption~A\ref{a:4} in which the coefficients $\beta^a_{r,\,j}$ do not vary with $r$ within $\mathcal D(a)$ (and vanish outside), so it implies Assumptions~A\ref{a:mu} and~A\ref{a:4}. For an isolated unit the coefficients $B_1^a,\,B_2^a$ multiply identically zero quantities and are thus unidentifiable; we fix them at zero. With this convention, $\mu_a(x,\,y) = B_0^a + B_1^ax + B_2^axy$ and the unit-level effect is $\gamma_a := \mu_a(1,\,1)-\mu_a(1,\,0) = B_2^a$ for every $a\in\mathcal A$.

Let $\phi_a$ denote the EARL weight of unit $a$,
\begin{equation}\label{def:phi}
    \phi_a(\mathbf S,\,\mathbf Z) = \sum_{r\in\mathcal D(a)}\frac{S_r}{q}\cdot\frac{Z_r - p}{p(1-p)},
\end{equation}
so that $\hat\tau^{\phantom{1}}_{EARL} = N^{-1}\sum_{a\in\mathcal A}\phi_a(\mathbf S,\,\mathbf Z)Y_a$, and write $\hat\gamma_a = \phi_aY_a$ for the unit-level estimate. The variance of EARL decomposes into pairwise covariances, $\var(\hat\tau_{EARL}) = N^{-2}\sum_{a_1,a_2}\cov(\hat\gamma_{a_1},\hat\gamma_{a_2})$, and we estimate each covariance by a weighted product of the observed outcomes, $Y_{a_1}Y_{a_2}R_{a_1,a_2}$. Denote by
$$
\mathcal P = \{(a_1,\,a_2)\in\mathcal A^2\colon \mathcal D(a_1)\cap\mathcal D(a_2)\neq\emptyset\}
$$
the set of \emph{overlapping} pairs (including the diagonal). The weights $R_{a_1,\,a_2}$ are constructed as follows.
\begin{itemize}
    \item If $(a_1,\,a_2)\notin\mathcal P$, set $R_{a_1,\,a_2}=0$ (under Assumptions A\ref{a:1}, A\ref{a:mu}, and A\ref{a:no-autocorrelation}, such pairs have zero covariance).
    \item If $(a_1,\,a_2)\in\mathcal P$ and $a_1\neq a_2$, let $\mathbf m_{a_1,a_2}$ be the vector of the $8$ monomials
    \begin{equation*}
        (G_{a_1},\,F_{a_1},\,G_{a_2},\,F_{a_2},\,G_{a_1}G_{a_2},\,G_{a_1}F_{a_2},\,F_{a_1}G_{a_2},\,F_{a_1}F_{a_2}),
    \end{equation*}
    let $\Sigma_{a_1,a_2}$ be its $8\times 8$ covariance matrix under the design, and let $\mathbf c_{a_1,a_2} = \Sigma^{-1}_{a_1,a_2}\mathbf u$, where $\mathbf u$ is the coordinate vector selecting $F_{a_1}F_{a_2}$. Set
    \begin{equation}\label{def:R}
        R_{a_1,\,a_2} = \phi_{a_1}\phi_{a_2} - \mathbf c_{a_1,a_2}^{\top}\left(\mathbf m_{a_1,a_2} - \E[\mathbf m_{a_1,a_2}]\right).
    \end{equation}
    \item If $a_1=a_2=a$, apply the same construction with the $5$-dimen\-sional monomial vector $\mathbf m_{a,a} = (G_a,\,F_a,\,G_a^2,\,G_aF_a,\,F_a^2)$ and $\mathbf u$ selecting $F_a^2$:
    $R_{a,\,a} = \phi_{a}^2 - \mathbf c_{a,a}^{\top}(\mathbf m_{a,a} - \E[\mathbf m_{a,a}])$.
\end{itemize}
The rationale mirrors \cite[Section~5.1]{harshaw2023erl}: the first term of $R_{a_1,a_2}$ makes $Y_{a_1}Y_{a_2}\phi_{a_1}\phi_{a_2}$ an (automatically unbiased) estimator of $\E[\hat\gamma_{a_1}\hat\gamma_{a_2}]$, while the linear-in-monomials term is calibrated so that $Y_{a_1}Y_{a_2}\,\mathbf c^{\top}(\mathbf m-\E[\mathbf m])$ is an unbiased estimator of the product $\gamma_{a_1}\gamma_{a_2}=\E[\hat\gamma_{a_1}]\E[\hat\gamma_{a_2}]$ for \emph{all} values of the coefficients $B^a_j$; subtracting the two yields an unbiased estimator of the covariance. All the moments entering $\E[\mathbf m_{a_1,a_2}]$ and $\Sigma_{a_1,a_2}$ are polynomials in $p$, $q$ determined by $|\mathcal D(a_1)|$, $|\mathcal D(a_2)|$, and $|\mathcal D(a_1)\cap\mathcal D(a_2)|$; they can be computed in closed form before the experiment or estimated by Monte Carlo \cite{fattorini2006applying}. Theorem~\ref{th:asymptotic-variance} below assumes the exact design moments; substituting Monte Carlo estimates forfeits exact unbiasedness, and asymptotic validity then requires the number of auxiliary draws (generated independently of the experiment) to grow fast enough that the resulting error in $\Sigma^{-1}_{a_1,a_2}$ remains asymptotically negligible. Note that only pairs in $\mathcal P$ require any computation, so the estimator scales with the number of overlapping pairs (at most $Nd_{\mathcal A}d_{\mathcal R}$) rather than $N^2$.

The construction requires the matrices $\Sigma_{a_1,a_2}$ to be invertible, which we assume in a uniform manner.

\begin{assumption}\label{a:lambda0}
    There exists a sequence $\lambda_N>0$ such that for every pair $(a_1,\,a_2)\in\mathcal P$ the smallest eigenvalue of $\Sigma_{a_1,a_2}$ is at least $\lambda_N$, and $d^{3}_{\mathcal A}d^{3}_{\mathcal R} = o(N\lambda_N^2)$.
\end{assumption}

Assumption~A\ref{a:lambda0} is the analogue of the non-degenerate-exposures condition of \cite[Assumption~4]{harshaw2023erl} and, like it, can be verified by the experimenter before the experiment, since $\Sigma_{a_1,a_2}$ depends only on the design. Allowing $\lambda_N\to0$ is essential: the smallest eigenvalue of a diagonal pair's $\Sigma_{a,a}$ cannot exceed its $(G_a,\,G_a)$ entry $\var(G_a) = q(1-q)/|\mathcal D(a)|$, so a lower bound independent of $N$ would force bounded analysis-side degrees. The eigenvalue condition fails entirely, for example, for analysis units with a single connection or for pairs with identical connection sets; for such pairs the corresponding covariance terms can be replaced by conservative upper bounds in the spirit of \cite[Proposition~5.2]{harshaw2023erl}, preserving asymptotic validity provided the number of affected \emph{ordered pairs}, multiplied by the magnitude of the replaced terms, is $o(N)$.

Given the weights, the variance estimator is
\begin{equation}\label{def:hat-V}
    \hat V = \frac 1 {N^2} \sum_{a_1\in\mathcal A}\sum_{a_2\in\mathcal A} Y_{a_1}Y_{a_2}R_{a_1,\,a_2}.
\end{equation}

The following theorem demonstrates that $\hat V$ is exactly unbiased and can be used to perform inference based on $\hat\tau_{EARL}^{\phantom{1}}$.

\begin{theorem}\label{th:asymptotic-variance}
    Let the assumptions of Theorem~\ref{th:asymptotic-normality} hold with Assumption~A\ref{a:4}$^*$ in place of Assumptions~A\ref{a:mu} and A\ref{a:4} (in particular, $d^{4}_{\mathcal{R}}d^{10}_{\mathcal{A}}=o(N)$ and $\E[Y_a^4]\le L$), and let Assumption~A\ref{a:lambda0} hold. Let the variance estimator $\hat V$ be defined by \eqref{def:R}--\eqref{def:hat-V}.

    Then $\E[\hat V] = \var(\hat\tau^{\phantom{1}}_{EARL})$ for every $N$, $\hat V / \var(\hat\tau^{\phantom{1}}_{EARL})\stackrel{p}{\to} 1$ as $N\to\infty$, and
    $$
        \frac{\hat\tau^{\phantom{1}}_{EARL} - GATE}{\sqrt{\hat V}} \stackrel{d}{\to} \mathcal N(0,\,1).
    $$
\end{theorem}

Like its full-assignment counterpart, $\hat V$ is not guaranteed to be non-negative in finite samples. Theorem~\ref{th:asymptotic-variance} implies $\Pr\{\hat V>0\}\to1$, and on the complementary event the studentised statistic can be defined arbitrarily (e.g., with $|\hat V|$ in place of $\hat V$), cf.\ \cite{harshaw2023erl}.

\subsubsection{Randomisation Inference}\label{subsubsec:ri}

An additional, computationally simpler variance estimator is available under a stronger null hypothesis. Specifically, consider
\begin{equation}\label{def:strong-null}
    H_0^*\colon\qquad e_a(\mathbf s,\,\mathbf z)\ \textup{does not depend on }(\mathbf s,\,\mathbf z),\quad \forall\,a\in\mathcal A.
\end{equation}
The null $H_0^*$ states that neither the assignment \emph{nor the allocation} of any randomisation unit affects the expected outcome of any analysis unit. It implies the pointwise null of zero expected effects, $\mu_a(1,\,1)-\mu_a(1,\,0)=0$ for all $a$, but is stronger. The strengthening is not incidental~--- it is necessitated by partial assignment. Under a concurrent experiment, the outcomes may co-vary with the allocation indicators $\mathbf S$ even when our treatment has no effect (unassigned units deliver a different experience); a re-randomisation procedure regenerates $(\mathbf S,\,\mathbf Z)$ independently of the observed outcomes and therefore destroys precisely this co-variation, so it can only reproduce the sampling variance when no such co-variation exists. Theorem~\ref{th:randomisation-inference} below uses $H_0^*$ together with Assumption~A\ref{a:no-autocorrelation}$^*$: constant conditional means and errors independent of the design together render the outcomes independent of $(\mathbf S,\,\mathbf Z)$. We demonstrate this effect numerically in Section~\ref{sec:simulations}.

The randomisation-inference methodology follows \cite[Section~4]{shi2025scalableanalysisbipartiteexperiments}: we re-randomise the allocation and assignment indicators, forming $J$ independent samples,
\begin{equation}\label{def:S-Z-samples}
    \begin{array}{c}
        \tilde{\mathbf{S}}_j = (\tilde S_{1,\,j},\ldots,\tilde S_{M,\,j}) \quad \iidsim \quad \mathrm{Bernoulli}(q),\\
        \tilde{\mathbf Z}_j = (\tilde Z_{1,\,j},\ldots,\tilde Z_{M,\,j}) \quad \iidsim \quad \mathrm{Bernoulli}(p),\\
        \tilde{\mathbf{S}}_j\ind \tilde{\mathbf Z}_j,
    \end{array}\qquad j=1,\,\ldots,\,J,
\end{equation}
drawn independently of the experimental data. Let $\hat\tau^{\phantom{1}}_{EARL}(\tilde {\mathbf S}_j,\,\tilde{\mathbf Z}_j)$ be the EARL estimator computed using the $j$-th generated sample of allocation and assignment indicators with the weight function \eqref{def:phi}, i.e.,
$$
    \hat\tau^{\phantom{1}}_{EARL}(\tilde {\mathbf S}_j,\,\tilde{\mathbf Z}_j) = \frac 1 N \sum_{a\in\mathcal{A}}\phi_a(\tilde {\mathbf S}_j,\,\tilde{\mathbf Z}_j)Y_a,\qquad j=1,\,\ldots,\,J.
$$
Note that these ``copies'' of the estimator are still computed using the same observed outcomes $Y_a$; only the allocation and assignment indicators are replaced with the re-randomised values.

Using the above notation, the randomisation-inference estimator of the variance of $\hat\tau^{\phantom{1}}_{EARL}$ is the empirical variance of the copies,
\begin{equation}\label{def:hatV}
    \hat V_{RI} = \frac{1}{J} \sum_{j=1}^J\left(\hat\tau^{\phantom{1}}_{EARL}(\tilde {\mathbf S}_j,\,\tilde{\mathbf Z}_j)-\frac 1 J\sum_{j'=1}^J\hat\tau^{\phantom{1}}_{EARL}(\tilde {\mathbf S}_{j'},\,\tilde{\mathbf Z}_{j'})\right)^2.
\end{equation}
We complete our theoretical development with a theorem demonstrating that $\hat V_{RI}$ can be used for significance testing based on $\hat\tau^{\phantom{1}}_{EARL}$. The validity proof for the analogous full-assignment procedure in \cite{shi2025scalableanalysisbipartiteexperiments} was found to have a flaw; the proof below (in Appendix~\ref{app:ri}) is self-contained.

\begin{theorem}\label{th:randomisation-inference}
    Let Assumptions A\ref{a:1}, A\ref{a:mu}, A\ref{a:no-autocorrelation}$^*$, and A\ref{a:non-degeneracy} hold, and suppose there exists $L>0$ such that $\E[Y_a^4] \le L$ for all $a\in\mathcal A$ (uniformly across all values of $N$). Let the strong null \eqref{def:strong-null} hold true, let $d^{4}_{\mathcal{R}}d^{10}_{\mathcal{A}}=o(N)$, and let the number of re-randomisations $J=J_N$ grow so that $d^{6}_{\mathcal{A}}d^{2}_{\mathcal{R}}=o(J_N)$. Let $\hat V_{RI}$ be defined by \eqref{def:S-Z-samples}--\eqref{def:hatV}.

    Then $\hat V_{RI} / \var(\hat\tau^{\phantom{1}}_{EARL})\stackrel{p}{\to} 1$ as $N\to\infty$, and
    for any $\alpha\in(0,\,1)$,
    \begin{equation*}
        \Pr\left\{\left|\hat\tau^{\phantom{1}}_{EARL}\right| > z_{1-\alpha/2}\sqrt{\hat V_{RI}}\right\}\to \alpha \qquad \textup{as }N\to\infty,
    \end{equation*}
    where $z_{1-\alpha/2}$ is the quantile of the standard normal distribution of level $1-\alpha/2$.
\end{theorem}

\section{Simulation Study}\label{sec:simulations}

We evaluate EARL in a semi-synthetic simulation study designed to stay close to the empirical evaluations of \citet{doudchenko2020causal} and \citet{harshaw2023erl}, while introducing partial assignment with allocation rates $q$ ranging from 10\% to 100\%. The full code, together with instructions reproducing every number and figure in this section, accompanies the paper.\footnote{The code is available at \url{https://github.com/akurennoy/earl}.} As in any simulation-based evaluation, the data-generating processes are necessarily stylised. We mitigate this in two ways: by building the experiment graphs from real datasets whose degree distributions are far from any convenient parametric form, and by including scenarios that \emph{violate} our assumptions. Nevertheless, the results should be read as controlled evidence about estimator behaviour with a known ground truth, which real experiments cannot provide, rather than as a substitute for field validation.

\subsection{Setup}\label{subsec:sim-setup}

\textbf{Graphs.} We use three bipartite graphs (Table~\ref{tab:graphs}). \emph{Synthetic}: the design of \cite[Section~6.1]{doudchenko2020causal}: $N=1000$ analysis units, $M=100$ randomisation units, each analysis unit connected to $m_a\sim U\{1,\ldots,10\}$ randomisation units chosen at random. \emph{Amazon}: the user--item graph built from the Amazon product-review data of \citet{he2016ups,mcauley2015image}, subsampling $N=1000$ users with between 2 and 50 reviews, similarly to \cite[Section~6.2]{doudchenko2020causal}; items are randomisation units, mirroring an experiment that changes an item-level mechanism (e.g., pricing) and measures customer satisfaction \cite{harshaw2023erl}. \emph{MovieLens}: the full MovieLens-100K user--movie graph \cite{harper2015movielens}, a deliberately dense stress case ($d_{\mathcal A}=737$) far outside our sparsity conditions.

\begin{table}
\caption{Experiment graphs used in the simulation study.}
\label{tab:graphs}
\begin{tabular}{lrrrrr}
\toprule
Graph & $N$ & $M$ & $|\mathcal C|$ & $d_{\mathcal A}$ & $d_{\mathcal R}$\\
\midrule
Synthetic & 1000 & 100 & 5415 & 10 & 76\\
Amazon & 1000 & 2718 & 3273 & 42 & 18\\
MovieLens & 943 & 1682 & 100000 & 737 & 583\\
\bottomrule
\end{tabular}
\end{table}

\textbf{Outcome models.} Outcomes follow $Y_a = e_a + \varepsilon_a$ with $\varepsilon_a\iidsiminline\mathcal N(0,\,\sigma_\varepsilon^2)$, $\sigma_\varepsilon^2 = 0.5$, as in \cite{doudchenko2020causal} (here and below the second argument of $\mathcal N$ denotes the variance) and
\begin{equation*}
    e_a = \alpha_a + \beta_a F_a + \gamma_u (1-G_a),
\end{equation*}
where $F_a$ is the joint coverage--exposure \eqref{eq:Fa} and the term $\gamma_u(1-G_a)$ captures the effect of unassigned connections (e.g., their participation in a concurrent experiment). The GATE equals $N^{-1}\sum_a\beta_a$. Following \cite[Section~8]{harshaw2023erl}, coefficients are drawn once and then held fixed across replications: \emph{S1} (positive heterogeneous effects) $\alpha_a\sim\mathcal N(-1,\,3/8)$, $\beta_a\sim\mathcal N(2,\,1)$; \emph{S2} (near-zero effects) $\alpha_a\sim\mathcal N(2,\,3/8)$, $\beta_a\sim\mathcal N(0,\,1/2)$; \emph{S3} (non-linear response, violating A\ref{a:4}) $e_a = \alpha_a + 4F_a(F_a-1)+\gamma_u(1-G_a)$ with $\alpha_a\sim\mathcal N(0,\,1/8)$ and $GATE=0$. Each scenario is run with $\gamma_u=1$ (unassigned units carry a concurrent-test effect), and S1 and S2 also with $\gamma_u=0$ (unassigned units deliver the control experience).

\textbf{Estimators.} We compare: \emph{EARL} \eqref{eq:earl}; \emph{ERL-drop} \eqref{eq:erl-drop}, the standard ERL on the reduced graph; \emph{IPW-assign}, the Horvitz--Thompson full-/empty-exposure contrast \cite{aronow2017estimating,zigler2021bipartite} computed on the reduced graph with assignment-only propensities $p^{k_a}$, $(1-p)^{k_a}$ ($k_a$ the number of participating connections); and \emph{IPW-alloc}, the same contrast on the full graph with propensities accounting for both allocation and assignment, $(qp)^{|\mathcal D(a)|}$ and $(q(1-p))^{|\mathcal D(a)|}$, in the spirit of the eligibility-aware propensities of \cite{tan2025estimating}. For each graph, scenario, and $q\in\{0.1,\,0.2,\ldots,\,1.0\}$ we run 1000 replications with $p=1/2$ and report the bias, the standard deviation, and the root mean squared error (RMSE) of each estimator.

\subsection{Results}\label{subsec:sim-results}

\begin{figure*}
\centering
\includegraphics[width=\textwidth]{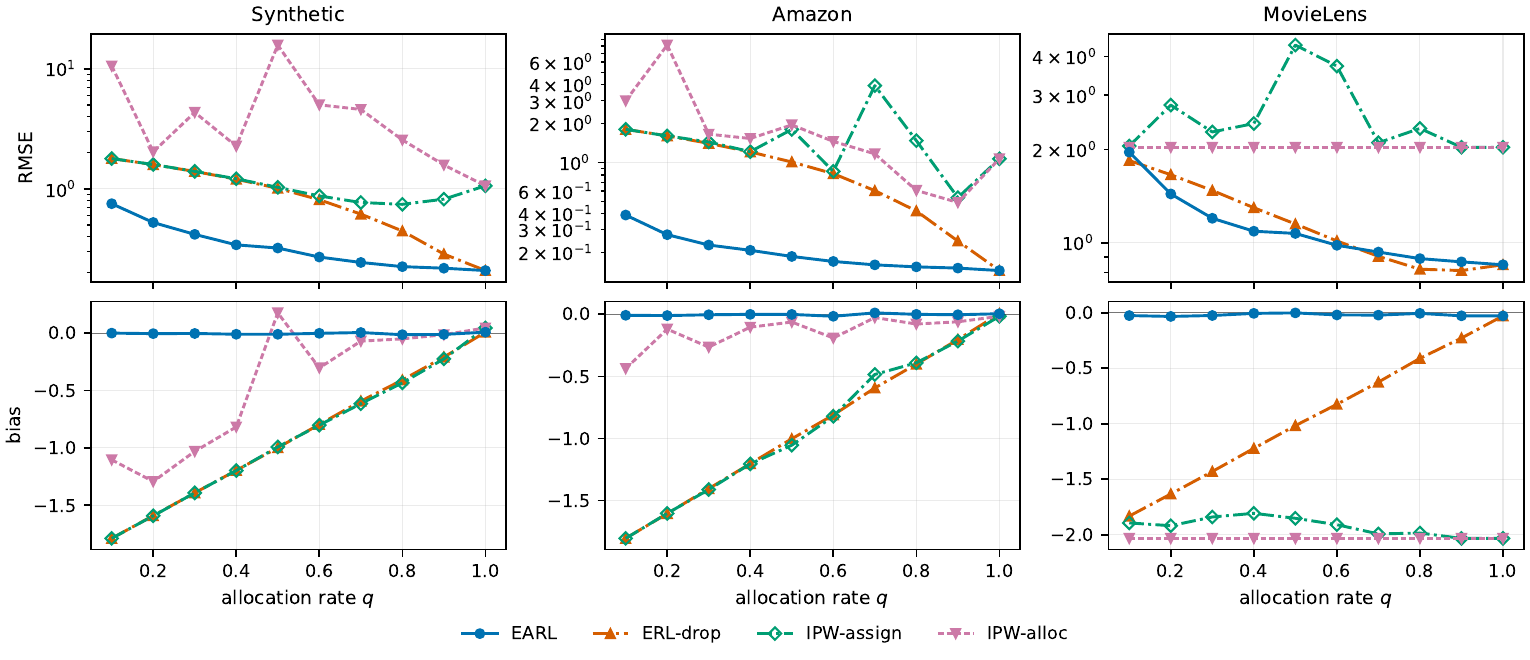}
\caption{RMSE (top, log scale) and bias (bottom) as functions of the allocation rate $q$ in scenario S1 ($\gamma_u=1$). EARL is unbiased everywhere and has the lowest RMSE at almost every allocation rate; ERL-drop has bias $-(1-q)\cdot GATE$; the IPW benchmarks are erratic (Synthetic, Amazon) or return zero in every replication (IPW-alloc on MovieLens, where a pure exposure was never observed). At $q=1$, EARL and ERL-drop coincide with the standard ERL estimator.}
\label{fig:rmse-bias}
\end{figure*}

\begin{table*}
\caption{RMSE (with bias in parentheses) at allocation rates $q\in\{0.2,\,0.5,\,0.8\}$, scenario S1 with a concurrent-test effect on unassigned units ($\gamma_u=1$); 1000 replications, $p=0.5$. Best RMSE per column in bold. GATE $\approx 2.0$ in all three graphs.}
\label{tab:rmse}
\centering
\small
\begin{tabular}{llccccccccc}
\toprule
& & \multicolumn{3}{c}{Synthetic} & \multicolumn{3}{c}{Amazon} & \multicolumn{3}{c}{MovieLens}\\
\cmidrule(lr){3-5}\cmidrule(lr){6-8}\cmidrule(lr){9-11}
Estimator & & $q=0.2$ & $q=0.5$ & $q=0.8$ & $q=0.2$ & $q=0.5$ & $q=0.8$ & $q=0.2$ & $q=0.5$ & $q=0.8$\\
\midrule
EARL & & \textbf{0.53} {\scriptsize$(-0.00)$} & \textbf{0.32} {\scriptsize$(-0.01)$} & \textbf{0.23} {\scriptsize$(-0.01)$} & \textbf{0.28} {\scriptsize$(-0.01)$} & \textbf{0.19} {\scriptsize$(-0.00)$} & \textbf{0.16} {\scriptsize$(-0.00)$} & \textbf{1.44} {\scriptsize$(-0.03)$} & \textbf{1.07} {\scriptsize$(-0.00)$} & 0.89 {\scriptsize$(-0.01)$}\\
ERL-drop & & 1.59 {\scriptsize$(-1.59)$} & 1.01 {\scriptsize$(-1.00)$} & 0.45 {\scriptsize$(-0.41)$} & 1.61 {\scriptsize$(-1.61)$} & 1.01 {\scriptsize$(-1.00)$} & 0.42 {\scriptsize$(-0.40)$} & 1.66 {\scriptsize$(-1.63)$} & 1.15 {\scriptsize$(-1.02)$} & \textbf{0.82} {\scriptsize$(-0.41)$}\\
IPW-assign & & 1.60 {\scriptsize$(-1.59)$} & 1.03 {\scriptsize$(-0.99)$} & 0.74 {\scriptsize$(-0.43)$} & 1.60 {\scriptsize$(-1.60)$} & 1.80 {\scriptsize$(-1.05)$} & 1.47 {\scriptsize$(-0.39)$} & 2.78 {\scriptsize$(-1.92)$} & 4.35 {\scriptsize$(-1.85)$} & 2.34 {\scriptsize$(-1.99)$}\\
IPW-alloc & & 2.04 {\scriptsize$(-1.29)$} & 15.79 {\scriptsize$(+0.18)$} & 2.56 {\scriptsize$(-0.05)$} & 8.13 {\scriptsize$(-0.12)$} & 1.95 {\scriptsize$(-0.06)$} & 0.61 {\scriptsize$(-0.08)$} & 2.03 {\scriptsize$(-2.03)$} & 2.03 {\scriptsize$(-2.03)$} & 2.03 {\scriptsize$(-2.03)$}\\
\bottomrule
\end{tabular}
\end{table*}

\textbf{Point estimation.} Figure~\ref{fig:rmse-bias} and Table~\ref{tab:rmse} summarise scenario S1, in which a genuine effect is present. Two findings stand out. First, EARL is unbiased at every allocation rate on all three graphs (largest absolute bias $0.03$, within Monte-Carlo error), while ERL-drop exhibits almost exactly the bias $-(1-q)\cdot GATE$ predicted by the theory~--- at $q=0.2$ it recovers barely a fifth of the true effect. IPW-assign inherits the same attenuation, because dropping unassigned units understates every unit's true connectivity, and IPW-alloc, though unbiased in principle, is so noisy that its empirical mean is still far from the target after 1000 replications. Second, EARL delivers the lowest RMSE in every cell of Table~\ref{tab:rmse} except one (on MovieLens at $q=0.8$ the biased ERL-drop edges it out), with the reduction relative to the strongest existing baseline reaching a factor of six on the Amazon graph at low allocation rates ($0.28$ vs $1.61$ at $q=0.2$) and an order of magnitude relative to the IPW benchmarks. The IPW estimators illustrate the variance pathology that motivates ERL-type methods: on MovieLens (degrees in the hundreds) a fully treated or fully untreated neighbourhood was never observed in any replication, so IPW-alloc returned $0$ throughout (its identifying exposure events have positive but astronomically small probability), yielding empirical bias $-GATE$ and zero empirical variance, and IPW-assign is nearly degenerate, while on sparser graphs their error is dominated by rare, enormous weights, and partial assignment amplifies the problem by shrinking the exposure propensities from $p^{|\mathcal D(a)|}$ to $(qp)^{|\mathcal D(a)|}$.

The $\gamma_u=0$ runs reveal the price of the $1/q$ reweighting. Since $\hat\tau_{EARL} = \hat\tau_{drop}/q$ under Bernoulli allocation, EARL's standard deviation is exactly $1/q$ times that of ERL-drop, and whether unbiasedness wins in RMSE depends on how this variance premium compares with the removed bias $(1-q)\cdot GATE$. The variance of both estimators is driven largely by the \emph{level} of the outcomes (the weights multiply $Y_a$, not a centred version of it): in S1 with $\gamma_u=1$ the concurrent-test term happens to offset the negative outcome intercept ($\E[\alpha_a]=-1$), so the level is small and EARL dominates broadly, as reported above; with $\gamma_u=0$ the level is of the order of the GATE itself, and on the synthetic ($q\le 0.2$) and MovieLens (all $q$) graphs the reduced-graph estimator's RMSE falls below EARL's, while on Amazon EARL still wins at every $q<1$ (by up to a factor of five). EARL remains unbiased throughout. The practical implication is standard: like all linear reweighting estimators, EARL should preferably be applied to outcomes centred by a fixed or pre-treatment baseline, or covariate-adjusted through a formally justified procedure (cf.\ CA-ERL \cite{shi2025scalableanalysisbipartiteexperiments}) and when the outcome level cannot be reduced and $q$ is very small, the experimenter faces a genuine bias--variance choice. Our closed-form variance estimator makes this choice quantifiable from the experiment itself, before acting on the estimate: since $\hat\tau_{drop}=q\,\hat\tau_{EARL}$, the two mean squared errors can be compared via the plug-in quantities $\hat V$ and $q^2\hat V + (1-q)^2\hat\tau^2_{EARL}$ (the latter overestimates its target by $(1-q)^2\var(\hat\tau_{EARL})$ in expectation; substituting $\hat\tau^2_{EARL}-\hat V$ removes the excess at the price of possible negativity), so the shrunk estimator is indicated when the estimated effect is small relative to EARL's sampling noise.

In scenario S2 (effects indistinguishable from zero, $GATE\approx0.02$) the ranking reverses: the shrinkage bias of ERL-drop is \emph{multiplicative} in the GATE, so with nothing to attenuate it attains the lowest RMSE among the non-degenerate estimators (on MovieLens even the identically-zero IPW-alloc trivially ``wins''), whereas EARL pays the variance cost of unbiasedness. We report this for transparency and add two remarks. First, the relative bias of ERL-drop is $100(1-q)\%$ of the estimand whatever its size, so the regime in which it wins on RMSE is that of near-zero effects; whenever the effect is consequential (because it is large, or because even a small effect matters at the operating scale) ERL-drop understates it by the factor $q$. Second, using the reduced graph is rarely a deliberate choice: an analysis pipeline that only sees the enrolled units reports $\hat\tau_{drop}$ with a confidence interval tightly centred on $q\cdot GATE$, giving the practitioner no indication that a correction is needed, nor by how much; and once partial assignment \emph{is} modelled, so that $q$ and the full graph are known, EARL is available at no extra cost. In scenario S3 (non-linear response) the linear-response guarantees behind EARL and ERL-drop no longer apply and both exhibit visible bias, while the IPW contrasts remain unbiased in principle but are far too unstable to benefit from it; EARL's error decreases towards the ERL error as $q\to1$, and no method dominates. Full tables for S2, S3, and the $\gamma_u=0$ configurations are reproducible from the accompanying code.

\begin{table}
\caption{Randomisation inference under the null of zero effects: empirical size of the nominal 5\% test and the ratio of the mean $\hat V_{RI}$ to the true variance of $\hat\tau_{EARL}$; 500 replications, $J=200$. When unassigned units carry a concurrent-test effect ($\gamma_u=1$), the strong null \eqref{def:strong-null} is violated even though all treatment effects are zero, and $\hat V_{RI}$ drifts away from the true variance.}
\label{tab:ri}
\begin{tabular}{llcccc}
\toprule
Graph & $\gamma_u$ & $q$ & Size & $\overline{\hat V}_{RI}/\var(\hat\tau)$\\
\midrule
Synthetic & 0 & 0.3 & 0.050 & 1.01\\
Synthetic & 0 & 0.7 & 0.060 & 0.99\\
Synthetic & 1 & 0.3 & 0.040 & 1.09\\
Synthetic & 1 & 0.7 & 0.040 & 1.17\\
Amazon & 0 & 0.3 & 0.042 & 1.08\\
Amazon & 0 & 0.7 & 0.052 & 0.95\\
Amazon & 1 & 0.3 & 0.046 & 1.05\\
Amazon & 1 & 0.7 & 0.050 & 1.02\\
\bottomrule
\end{tabular}
\end{table}

\textbf{Randomisation inference.} Table~\ref{tab:ri} evaluates the randomisation-inference procedure of Section~\ref{subsubsec:ri} under a data-generating process with all unit-level treatment effects equal to zero ($\beta_a\equiv 0$). When the strong null $H_0^*$ holds exactly ($\gamma_u=0$), the empirical size of the nominal 5\% test is within Monte-Carlo error of the nominal level and $\hat V_{RI}$ matches the true variance (ratios $0.95$--$1.08$). When unassigned units carry a concurrent-test effect ($\gamma_u=1$)~--- so that expected outcomes depend on $\mathbf S$ although all treatment effects are zero~--- the variance ratio drifts up to $1.17$ on the synthetic graph, illustrating the theoretical point of Section~\ref{subsubsec:ri}: re-randomisation reproduces the estimator's variance only under the allocation-inclusive null. In this particular data-generating process the drift is conservative (the test under-rejects), but its direction is not guaranteed in general, which is why Theorem~\ref{th:randomisation-inference} requires $H_0^*$; the closed-form estimator $\hat V$ of Section~\ref{subsubsec:variance} requires no null hypothesis, although it does rest on the homogeneous linear conditional-mean model and the other conditions of Theorem~\ref{th:asymptotic-variance}.

\section{Conclusion}\label{sec:conclusion}

We studied bipartite experiments in which only part of the randomisation-unit population is enrolled~--- a common regime on experimentation platforms that cap rollouts and run concurrent tests. The standard practice of restricting the analysis to the enrolled units attenuates the estimated global treatment effect. To address this problem, we proposed EARL, an exposure- and allocation-reweighted linear estimator that removes this bias by upweighting every participating connection by its inverse allocation propensity. EARL is provably unbiased whatever experience the unassigned units receive (so long as they contribute to expected outcomes additively), minimum-variance in its class of linear reweighting estimators, and consistent and asymptotically normal under the same graph-sparsity conditions known for full-assignment ERL; it is accompanied by two variance estimators: an exactly unbiased closed-form estimator and a randomisation-inference procedure, for which we provide validity theory and identify the allocation-inclusive null hypothesis it requires. Simulations on real bipartite graphs confirm the theory, show up to six-fold accuracy gains over the strongest existing baseline (in some configurations more than an order of magnitude over Horvitz--Thompson alternatives), and delineate the low-allocation, high-outcome-level regime where the variance premium of unbiasedness makes outcome centring advisable.

Several directions remain open. Natural extensions include unit-specific allocation and assignment probabilities, covariate adjustment in the spirit of CA-ERL \cite{shi2025scalableanalysisbipartiteexperiments}, weighted experiment graphs, and designs that optimise the allocation split across concurrent experiments jointly with the exposure distribution \cite{harshaw2023erl,brennan2022cluster}. On the theoretical side, characterising the efficiency frontier beyond the linear reweighting class, and relaxing the homogeneity assumption underlying variance estimation, both appear fruitful.

\bibliographystyle{ACM-Reference-Format}
\bibliography{references}

\appendix

\section{Proofs and Auxiliary Statements}\label{app:proofs}

Throughout the appendix we use the shorthand
\begin{equation}\label{def:Kw}
    \tilde Z_r = \frac{Z_r-p}{p(1-p)},\qquad
    K_w = \frac{d_{\mathcal A}}{q\,p(1-p)},
\end{equation}
so that the EARL weight \eqref{def:phi} is $\phi_a = \sum_{r\in\mathcal D(a)}(S_r/q)\tilde Z_r$; the constant $K_w$, whose subscript stands for ``weight'', is a uniform bound on these weights (Lemma~\ref{lem:weight-bound}). We write $\hat\gamma_a=\phi_aY_a$ for the unit-level terms of the estimator, so that $\hat\tau_{EARL}=N^{-1}\sum_{a\in\mathcal A}\hat\gamma_a$, and
\begin{equation}\label{def:gamma-a}
    \gamma_a = \mu_a(1,\,1) - \mu_a(1,\,0)
\end{equation}
for the unit-level estimands. We will also repeatedly use the set
\begin{equation*}
    \mathcal T(a) = \{a'\in\mathcal A\mid \mathcal D(a')\cap \mathcal D(a) \neq \emptyset\}
\end{equation*}
of analysis units sharing a connected randomisation unit with $a$; note $|\mathcal T(a)|\le d_{\mathcal A}d_{\mathcal R}$ and $a\in\mathcal T(a)$ whenever $\mathcal D(a)\neq\emptyset$.

\subsection{Unbiasedness of the EARL Estimator}

\begin{proof}[Proof of Theorem~\ref{th:unbiasedness}]
Step 1. We begin by showing that
\begin{equation}\label{eq:EhatG=}
    \E[\hat \gamma_a] = \frac 1 q \cdot\frac{\partial\mu_a}{\partial y}(q,\,p)\qquad\forall\,a\in\mathcal A.
\end{equation}
For this purpose, consider the following partial Horvitz--Thompson reweighting of the observation $Y_a$, which tilts the assignment rate of the \emph{participating} units from $p$ to $y$ while leaving the allocation untouched:
\begin{equation*}
    T_a(y) = Y_a\prod_{r\in\mathcal{D}(a)}\left[S_r\left(\frac{y}{p}\right)^{Z_r}\left(\frac{1-y}{1-p}\right)^{1-Z_r} + (1-S_r)\right],\ \  y\in(0,\,1),
\end{equation*}
where a product over an empty set is defined to be one. We claim that
\begin{equation}\label{eq:HTunb}
    \E[T_a(y)] = \mu_a(q,\,y)\qquad \forall\,y\in(0,\,1).
\end{equation}
Indeed, by the law of iterated expectations, $\E[T_a(y)]$ equals the expectation of $e_a(\mathbf S,\,\mathbf Z)$ multiplied by the product in the display. Conditionally on $\mathbf S$, taking the expectation over $\mathbf Z$ replaces the distribution of $Z_r$ by $\mathrm{Bernoulli}(y)$ for every $r$ with $S_r=1$ (the factor $(y/p)^{z_r}((1-y)/(1-p))^{1-z_r}$ is exactly the likelihood ratio between $\mathrm{Bernoulli}(y)$ and $\mathrm{Bernoulli}(p)$), while the coordinates with $S_r=0$ keep the distribution $\mathrm{Bernoulli}(p)$. By Assumption~A\ref{a:mu}(ii), $e_a$ does not depend on the assignment coordinates of units with $S_r=0$, so the latter coordinates can equivalently be drawn from $\mathrm{Bernoulli}(y)$ without changing the expectation; by Assumption~A\ref{a:mu}(i), neither do the coordinates outside $\mathcal D(a)$ matter. Hence
\begin{equation*}
    \E[T_a(y)\mid \mathbf S] = \E_{Z_r\iidsim \mathrm{Bern}(y)}\left[e_a(\mathbf S,\,\mathbf Z)\right],
\end{equation*}
and averaging over $\mathbf S\sim \mathrm{Bernoulli}(q)^{\otimes M}$ yields \eqref{eq:HTunb} by the definition \eqref{eq:ma} of $\mu_a$.

For any $\omega\in\Omega$ the map $y\mapsto T_a(y)$ is a polynomial and therefore differentiable, and for all $y\in(0,\,1)$,
\begin{equation*}
    \left|\frac{\partial T_a}{\partial y}(y)\right|
    \le |Y_a|\, d_{\mathcal A}\left(\max\left(\frac 1p,\,\frac 1{1-p}\right)\right)^{d_{\mathcal A}}
\end{equation*}
(each bracketed factor and each of its derivatives in $y$ is bounded by $\max(1/p,\,1/(1-p))$), where the right-hand side is an integrable random variable. Thus, we can swap differentiation and integration, i.e.,
\begin{equation}\label{eq:swap}
    \E\left[\frac{\partial T_a}{\partial y}(y)\right] = \frac{\partial}{\partial y}\E[T_a(y)] \stackrel{\eqref{eq:HTunb}}{=} \frac{\partial \mu_a}{\partial y}(q,\,y),\quad y\in(0,\,1).
\end{equation}
Now, by direct computation (each bracketed factor equals $1$ at $y=p$, and its derivative at $y=p$ equals $S_r\tilde Z_r$),
\begin{equation}\label{eq:d2HT=}
    \frac{\partial T_a}{\partial y}(p) =
    \sum_{r\in\mathcal{D}(a)} S_r\tilde Z_r\, Y_a = q\,\phi_a Y_a = q\,\hat\gamma_a.
\end{equation}
The combination of \eqref{eq:swap} at $y=p$ and \eqref{eq:d2HT=} gives \eqref{eq:EhatG=}.

Step 2. The next step is to prove that
\begin{equation}\label{eq:GATE=d2mu}
    \exists\,\eta,\,\xi\in(0,\,1)\colon\quad GATE = \frac{\partial^2\mu}{\partial x \partial y}(\eta,\,\xi).
\end{equation}
From definitions \eqref{eq:ma} and \eqref{eq:mu}, it follows that the function $\mu$ is a polynomial and hence $\mu(1,\,\cdot)$ is continuous on $[0,\,1]$ and differentiable on $(0,\,1)$, satisfying the mean-value theorem \cite[Theorem~5.10]{rudin1976principles}. Thus, there exists $\xi\in(0,\,1)$ such that
\begin{equation}\label{eq:mv-mu}
    \mu(1,\,1)-\mu(1,\,0) = \frac{\partial\mu}{\partial y}(1,\,\xi).
\end{equation}
Next, observe that
\begin{equation}\label{eq:mu0const}
    \mu(0,\,\cdot) = const,
\end{equation}
because if none of the randomisation units participates in the experiment, the expected outcome of any analysis unit does not depend on the rate at which participating units would be assigned to treatment (formally, this is Assumption~A\ref{a:mu}(ii) applied to $\mathbf s = \mathbf 0$ in \eqref{eq:ma}). Consequently,
\begin{equation*}
    \frac{\partial\mu}{\partial y}(0,\,\cdot) \equiv 0,\qquad\textup{in particular}\qquad
    \frac{\partial\mu}{\partial y}(0,\,\xi) = 0.
\end{equation*}
Using the last equality we can rewrite \eqref{eq:mv-mu} as
\begin{equation*}
    \mu(1,\,1)-\mu(1,\,0) = \frac{\partial\mu}{\partial y}(1,\,\xi) - \frac{\partial\mu}{\partial y}(0,\,\xi).
\end{equation*}
Applying the mean-value theorem once again (to the function $x\mapsto \frac{\partial\mu}{\partial y}(x,\,\xi)$), we obtain that
\begin{equation}\label{eq:deltamu=d2dxdy}
    \mu(1,\,1)-\mu(1,\,0) = \frac{\partial^2\mu}{\partial x\partial y}(\eta,\,\xi)
\end{equation}
for some $\eta\in(0,\,1)$. The definition of GATE \eqref{eq:gate} and the relation \eqref{eq:deltamu=d2dxdy} imply \eqref{eq:GATE=d2mu}.

Step 3. By the same argument applied to the function $x\mapsto\frac{\partial \mu}{\partial y}(x,\,p)$ on the interval $[0,\,q]$ (whose value at $0$ vanishes by \eqref{eq:mu0const}), averaging \eqref{eq:EhatG=} over $a\in\mathcal A$ yields
\begin{equation}\label{eq:EhatTau}
    \E[\hat\tau^{\phantom{1}}_{EARL}] = \frac 1q\cdot \frac{\partial\mu}{\partial y}(q,\,p) = \frac 1q \cdot q\cdot\frac{\partial^2\mu}{\partial x\partial y}(\tilde x,\,p) = \frac{\partial^2\mu}{\partial x\partial y}(\tilde x,\,p)
\end{equation}
for some $\tilde x\in(0,\,q)$.

Step 4. By combining \eqref{eq:GATE=d2mu} with \eqref{eq:EhatTau} and applying the mean-value theorem to the function $t\mapsto \frac{\partial^2\mu}{\partial x\partial y}\left((\tilde x,\,p) + t\,(\eta-\tilde x,\,\xi - p)\right)$ on $[0,\,1]$, we find that there exists $t\in(0,\,1)$ such that, with $(x_t,\,y_t) = (\tilde x,\,p) + t(\eta - \tilde x,\,\xi - p)$,
\begin{eqnarray*}
    \E[\hat\tau_{EARL}] - GATE &=& \frac{\partial^3\mu}{\partial x^2\partial y}(x_t,\,y_t)\,(\tilde x-\eta) \nonumber\\ &&{}+ \frac{\partial^3\mu}{\partial x\partial y^2}(x_t,\,y_t)\,(p-\xi).
\end{eqnarray*}
Since $|\tilde x-\eta| < 1$ and $|p-\xi| < 1$, the bound \eqref{eq:bound} readily follows.

Step 5. Finally, suppose that Assumption~A\ref{a:4} holds true. In that case, for any analysis unit $a\in\mathcal A$ the expected outcome function $\mu_a$ defined in \eqref{eq:ma} is a quadratic polynomial,
\begin{eqnarray*}
    \mu_a(x,\,y) &\stackrel{\eqref{eq:ma}}{=}& \E_{S_r\iidsim \mathrm{Bern}(x),\,Z_r\iidsim \mathrm{Bern}(y),\,\mathbf S\ind \mathbf Z}[e_a(\mathbf S,\,\mathbf Z)] \\
    &\stackrel{A\ref{a:4}}{=}& \sum_{r=1}^M(\beta^a_{r,\,0} + \beta^a_{r,\,1}(1-x) + \beta^a_{r,\,2}\, xy + \beta^a_{r,\,3}\,x(1-y)),
\end{eqnarray*}
which implies that the average expected outcome $\mu$ defined in \eqref{eq:mu} is also a quadratic polynomial of this form, and hence $\frac{\partial^3\mu}{\partial x^2\partial y}$ and $\frac{\partial^3\mu}{\partial x\partial y^2}$ are zero everywhere on $(0,\,1)\times (0,\,1)$. Therefore, the right-hand side of \eqref{eq:bound} is zero, implying \eqref{eq:unbiasedness}.
\end{proof}

\begin{remark}\label{rmk:unbiasedness}
    \textup{
    From the proof of Theorem~\ref{th:unbiasedness} it can be seen that under Assumptions~A\ref{a:1}--A\ref{a:4}, each individual $\hat\gamma_a$ is unbiased for its respective estimand $\gamma_a$ \eqref{def:gamma-a}. Indeed, in Step 1 we establish that $\E[\hat\gamma_a] = q^{-1}\frac{\partial\mu_a}{\partial y}(q,\,p)$. The logic of Steps 2--3 applies to each individual function $\mu_a$ and shows that $\E[\hat\gamma_a] = \frac{\partial^2\mu_a}{\partial x \partial y}(\tilde x_a,\,p)$ and $\gamma_a = \frac{\partial^2\mu_a}{\partial x \partial y}(\eta_a,\,\xi_a)$ for some interior points. Finally, Step 5 uses Assumption~A\ref{a:4} to show that $\mu_a$ is a quadratic polynomial with constant mixed derivative, hence $\E[\hat\gamma_a] = \gamma_a$.}
\end{remark}

\subsection{Optimality within the Linear Reweighting Class}\label{app:optimality}

We begin with a lemma that is used in the proof of Proposition~\ref{prop:optimality} below and in several later proofs; note that its two parts require different subsets of the assumptions.

\begin{lemma}\label{lem:cov=0}
    Let the random variables $S_r,\,Z_r\colon\Omega\mapsto\{0,\,1\}$, $r=1,\,\ldots,\,M$, satisfy Assumption A\ref{a:1}.

    Then
    \begin{itemize}
        \item[\textup{(a)}] if Assumption~A\ref{a:no-autocorrelation} holds true (in particular, the outcome variables are square integrable), then
            $\E[Y_{a_1}Y_{a_2}\mid \mathbf{S},\,\mathbf{Z}] = e_{a_1}(\mathbf{S},\,\mathbf{Z})\,e_{a_2}(\mathbf{S},\,\mathbf{Z})$ for all $a_1\neq a_2\in\mathcal A$;
        \item[\textup{(b)}] if, in addition, Assumptions~A\ref{a:mu} and A\ref{a:4} hold true, then for any two analysis units $a_1,\,a_2\in\mathcal{A}$ such that $\mathcal D(a_1)\cap \mathcal D(a_2) = \emptyset$ it holds that
        \begin{equation}\label{eq:cov=0}
            \E[(\hat\gamma_{a_1} - \gamma_{a_1})(\hat\gamma_{a_2} - \gamma_{a_2})] = 0.
        \end{equation}
    \end{itemize}
\end{lemma}
\begin{proof}
    Let $\varepsilon_a$, $a\in\mathcal{A}$, be the errors defined by \eqref{def:epsa}. From the definition, it follows that
    \begin{eqnarray}\label{eq:cond-E-eps-a=0}
        \E[\varepsilon_a\mid \mathbf{S},\,\mathbf{Z}] = 0\qquad\forall\,a\in\mathcal A.
    \end{eqnarray}
    Furthermore, due to Assumption~A\ref{a:no-autocorrelation} we have that for $a_1\neq a_2$,
    \begin{eqnarray}\label{eq:E-prod-eps=0}
        \E[\varepsilon_{a_1}\varepsilon_{a_2}\mid \mathbf{S},\,\mathbf{Z}]
        &\stackrel{A\ref{a:no-autocorrelation}}{=}& \E[\varepsilon_{a_1}\mid \mathbf{S},\,\mathbf{Z}]\cdot
        \E[\varepsilon_{a_2}\mid \mathbf{S},\,\mathbf{Z}]
        \stackrel{\eqref{eq:cond-E-eps-a=0}}{=} 0.
    \end{eqnarray}
    Using \eqref{eq:cond-E-eps-a=0} and \eqref{eq:E-prod-eps=0}, we obtain that
    \begin{eqnarray}\label{eq:E-prod-Y=E-prod-e}
        \E[Y_{a_1}Y_{a_2}\mid \mathbf{S},\,\mathbf{Z}] &=& \E[
        (e_{a_1}(\mathbf{S},\,\mathbf{Z}) + \varepsilon_{a_1})
        (e_{a_2}(\mathbf{S},\,\mathbf{Z}) + \varepsilon_{a_2})\mid \mathbf{S},\,\mathbf{Z}
        ]\nonumber\\
        &=& e_{a_1}(\mathbf{S},\,\mathbf{Z})\,e_{a_2}(\mathbf{S},\,\mathbf{Z}),
    \end{eqnarray}
    which proves (a).

    For part (b), consider any two analysis units $a_1$ and $a_2$ with $\mathcal{D}(a_1)\cap \mathcal{D}(a_2)=\emptyset$. Given Assumption~A\ref{a:mu} and the definition of $\phi_a$ \eqref{def:phi}, the random variables $\phi_{a_1} e_{a_1}(\mathbf{S},\,\mathbf{Z})$ and $\phi_{a_2}e_{a_2}(\mathbf{S},\,\mathbf{Z})$ are functions of non-overlapping subsets of allocation and assignment indicators and are therefore independent due to Assumption~A\ref{a:1}. Using this observation together with \eqref{eq:E-prod-Y=E-prod-e}, we conclude that
    \begin{eqnarray}\label{eq:=0}
        \E[\hat\gamma_{a_1}\hat\gamma_{a_2}]
        &=& \E\left[\phi_{a_1}\phi_{a_2}\,\E[Y_{a_1}Y_{a_2}\mid \mathbf{S},\,\mathbf{Z}]\right]
        \nonumber\\
        &\stackrel{\eqref{eq:E-prod-Y=E-prod-e}}{=}& \E\left[\phi_{a_1}e_{a_1}(\mathbf{S},\,\mathbf{Z})\,\phi_{a_2}e_{a_2}(\mathbf{S},\,\mathbf{Z})\right]\nonumber\\
        &=& \E\left[\phi_{a_1}e_{a_1}(\mathbf{S},\,\mathbf{Z})\right]\cdot\E[\phi_{a_2}e_{a_2}(\mathbf{S},\,\mathbf{Z})] \nonumber\\
        &=& \E[\hat\gamma_{a_1}]\,\E[\hat\gamma_{a_2}].
    \end{eqnarray}
    The last equality and Remark~\ref{rmk:unbiasedness} imply \eqref{eq:cov=0}, completing the proof.
\end{proof}

\begin{proof}[Proof of Proposition~\ref{prop:optimality}]
    Write $\tilde Z_r=(Z_r-p)/(p(1-p))$ and note $\E[\tilde Z_r]=0$, $\E[\tilde Z_rZ_r]=1$. Any $c\colon\{0,\,1\}\mapsto\R$ can be written as
    \begin{equation*}
        c(s) = c(1)\,s + c(0)\,(1-s) = q\,c(1)\cdot\frac sq + c(0)\,(1-s),
    \end{equation*}
    so that
    \begin{equation}\label{eq:c-decomp}
        \hat\tau_c = q\,c(1)\,\hat\tau_{EARL} + c(0)\,\hat\tau_U
    \end{equation}
    with $\hat\tau_U$ as defined in the proposition.

    We first show that
    \begin{equation}\label{eq:U-orth}
        \E[\hat\tau_U] = 0\quad\textup{and}\quad \E[\hat\tau_U\,\hat\tau_{EARL}] = 0
    \end{equation}
    for every outcome model satisfying the stated conditions. Fix a randomisation unit $r$ and let $\mathbf Z_{-r}$ collect all assignment coordinates except $Z_r$. Under Assumption~A\ref{a:1}, $Z_r$ is independent of $(\mathbf S,\,\mathbf Z_{-r})$, so for every function $f$ with $\E[|f(\mathbf S,\,\mathbf Z_{-r})|]<\infty$,
    \begin{eqnarray}\label{eq:kill-term}
        \E\left[(1-S_r)\tilde Z_r\, f(\mathbf S,\,\mathbf Z_{-r})\right] &=& \E[\tilde Z_r]\cdot\E\left[(1-S_r)\, f(\mathbf S,\,\mathbf Z_{-r})\right]\nn\\
        &=& 0.
    \end{eqnarray}
    We bring every summand of $\E[\hat\tau_U]$ and $\E[\hat\tau_U\hat\tau_{EARL}]$ into this form by replacing the outcomes with their conditional moments given $(\mathbf S,\,\mathbf Z)$ and then removing $Z_r$ from those moments. The removal step relies on the following observation: by Assumption~A\ref{a:mu}(ii), $e_a(\mathbf s,\,\mathbf z)$ does not depend on $z_r$ when $s_r=0$, so there exists an integrable function $h_{a,r}$ such that
    \begin{equation}\label{eq:h-subst}
        (1-S_r)\,e_a(\mathbf S,\,\mathbf Z) = (1-S_r)\,h_{a,r}(\mathbf S,\,\mathbf Z_{-r})\quad\textup{a.s.};
    \end{equation}
    analogously, the second-moment condition of the proposition provides an integrable function $g_{a,r}$ with $(1-S_r)\,m_a(\mathbf S,\,\mathbf Z) = (1-S_r)\,g_{a,r}(\mathbf S,\,\mathbf Z_{-r})$ a.s. In the displays below we suppress the arguments of $e_a$, $m_a$, $h_{a,r}$, and $g_{a,r}$.

    A summand of $\E[\hat\tau_U]$ has the form $\E[(1-S_r)\tilde Z_rY_a]$ with $r\in\mathcal D(a)$. Since $(1-S_r)\tilde Z_r$ is a function of $(\mathbf S,\,\mathbf Z)$, conditioning on $(\mathbf S,\,\mathbf Z)$ gives
    \begin{eqnarray*}
        \E[(1-S_r)\tilde Z_rY_a] &=& \E[(1-S_r)\tilde Z_r\,e_a]
        \stackrel{\eqref{eq:h-subst}}{=} \E[(1-S_r)\tilde Z_r\,h_{a,r}]\\
        &\stackrel{\eqref{eq:kill-term}}{=}& 0,
    \end{eqnarray*}
    proving $\E[\hat\tau_U]=0$. A summand of $\E[\hat\tau_U\hat\tau_{EARL}]$ has the form $\E[(1-S_r)\tilde Z_r\, (S_{r'}/q)\tilde Z_{r'}\, Y_aY_{a'}]$ with $r\in\mathcal D(a)$ and $r'\in\mathcal D(a')$. If $r'=r$, it vanishes because $(1-S_r)S_r=0$, so let $r'\neq r$; then $(S_{r'}/q)\tilde Z_{r'}$ is a bounded function of $(\mathbf S,\,\mathbf Z_{-r})$. For $a\neq a'$, conditioning on $(\mathbf S,\,\mathbf Z)$ and applying Lemma~\ref{lem:cov=0}(a),
    \begin{eqnarray*}
        \lefteqn{\E\left[(1-S_r)\tilde Z_r\, (S_{r'}/q)\tilde Z_{r'}\, Y_aY_{a'}\right]}\\
        &=& \E\left[(1-S_r)\tilde Z_r\, (S_{r'}/q)\tilde Z_{r'}\, e_a\,e_{a'}\right]\\
        &\stackrel{\eqref{eq:h-subst}}{=}& \E\left[(1-S_r)\tilde Z_r\, (S_{r'}/q)\tilde Z_{r'}\, h_{a,r}\,h_{a',r}\right]
        \stackrel{\eqref{eq:kill-term}}{=} 0,
    \end{eqnarray*}
    where the integrability required by \eqref{eq:kill-term} holds because, by the Cauchy--Schwarz and (conditional) Jensen inequalities, $\E[|e_ae_{a'}|]\le(\E[Y_a^2]\,\E[Y_{a'}^2])^{1/2}<\infty$. For $a=a'$, the same computation goes through with $\E[Y_a^2\mid\mathbf S,\,\mathbf Z]=m_a$ in place of Lemma~\ref{lem:cov=0}(a) and with $g_{a,r}$ in place of $h_{a,r}h_{a',r}$. Summing over all terms proves \eqref{eq:U-orth}.

    Part (a). By \eqref{eq:c-decomp} and \eqref{eq:U-orth}, $\E[\hat\tau_c] = q\,c(1)\,\E[\hat\tau_{EARL}]$ for every admissible outcome model. If $q\,c(1)=1$, the expectations agree always. Conversely, take the outcome model $Y_a = F_a$ (which satisfies all the stated conditions: $F_a$ is a bounded function of $\{(S_r,\,S_rZ_r)\colon r\in\mathcal D(a)\}$, so $\varepsilon_a=0$ and both $e_a=F_a$ and $m_a=F_a^2$ are local and exclude would-be assignments): then $\E[\hat\tau_{EARL}] = N^{-1}\sum_a \mathbf 1\{\mathcal D(a)\neq\emptyset\}\neq 0$ because the graph contains a connection, so $\E[\hat\tau_c]=\E[\hat\tau_{EARL}]$ forces $q\,c(1)=1$. The final claim of (a) follows from Theorem~\ref{th:unbiasedness}(b).

    Part (b). With $q\,c(1)=1$, decomposition \eqref{eq:c-decomp} reads $\hat\tau_c = \hat\tau_{EARL} + c(0)\hat\tau_U$, and by \eqref{eq:U-orth} the two terms are uncorrelated with $\E[\hat\tau_U]=0$. Hence
    \begin{equation*}
        \var(\hat\tau_c) = \var(\hat\tau_{EARL}) + c(0)^2\var(\hat\tau_U). \qedhere
    \end{equation*}
\end{proof}

\subsection{Consistency of the EARL Estimator}

We begin with the following simple lemma.
\begin{lemma}\label{lem:weight-bound}
    Let the random variables $S_r,\,Z_r\colon\Omega\mapsto\{0,\,1\}$, $r=1,\,\ldots,\,M$, satisfy Assumption A\ref{a:1} and let $\phi_a$ be defined as in \eqref{def:phi}. Then, with $K_w$ from \eqref{def:Kw},
    \begin{equation*}
        \left|\phi_a\right| \le K_w\qquad \forall\,a\in\mathcal A.
    \end{equation*}
\end{lemma}
\begin{proof}
    Each summand of $\phi_a$ satisfies $|(S_r/q)\tilde Z_r| \le \frac{\max(p,\,1-p)}{q\,p(1-p)} \le \frac{1}{q\,p(1-p)}$, and there are at most $d_{\mathcal A}$ summands.
\end{proof}

Similarly to \cite[Lemma~A.3]{harshaw2023erl}, we proceed by establishing a bound on the $k$-th absolute moment of the error $(\hat\gamma_a - \gamma_a)$.
\begin{lemma}\label{lem:kth-abs-moment-bound}
    Let the random variables $S_r,\,Z_r\colon\Omega\mapsto\{0,\,1\}$, $r=1,\,\ldots,\,M$, satisfy Assumption A\ref{a:1}. Let the outcome variables $Y_a$, $a\in\mathcal{A}$, have a finite absolute moment of order $K\ge 1$ and let Assumptions A\ref{a:mu} and A\ref{a:4} be true.

    Then for each $k=1,\,\ldots,\,K$ it holds that
    \begin{equation*}
        \E[|\hat\gamma_a-\gamma_a|^k] \le \left(2K_w\right)^k \E[|Y_a|^k].
    \end{equation*}
\end{lemma}
\begin{proof}
    Using Lemma~\ref{lem:weight-bound}, we obtain that for any $k=1,\,\ldots,\,K$,
    \begin{eqnarray}\label{eq:E-hat-ga-k-bound}
        \E[|\hat\gamma_a|^k] &=& \E\left[\left|\phi_a\right|^k|Y_a|^k\right]
        \stackrel{\textup{Lemma}~\ref{lem:weight-bound}}{\le} K_w^k\,\E[|Y_a|^k].
    \end{eqnarray}
    From the last inequality and Remark~\ref{rmk:unbiasedness} it follows that
    \begin{eqnarray}\label{eq:ga-bound}
        |\gamma_a| = |\E[\hat\gamma_a]| \le \E[|\hat\gamma_a|] \le K_w\,\E[|Y_a|] \le K_w\left(\E[|Y_a|^k]\right)^{1/k},
    \end{eqnarray}
    where the last step is Jensen's inequality. Combining \eqref{eq:E-hat-ga-k-bound} and \eqref{eq:ga-bound}, and using Minkowski's inequality we derive that
    \begin{eqnarray*}
        \E[|\hat\gamma_a-\gamma_a|^k] &\le& \left(\left(\E[|\hat\gamma_a|^k]\right)^{1/k} + |\gamma_a|\right)^{k}\\
        &\le& \left(K_w\left(\E[|Y_a|^k]\right)^{1/k}+K_w\left(\E[|Y_a|^k]\right)^{1/k}\right)^k\\
        &=& \left(2K_w\right)^k\E[|Y_a|^k]
    \end{eqnarray*}
    for all $k=1,\,\ldots,\,K$.
\end{proof}

\begin{proof}[Proof of Theorem~\ref{th:consistency}]
    The proof is analogous to \cite[Theorem~4.2]{harshaw2023erl}. Consider the mean squared error of $\hat\tau_{EARL}$,
    \begin{eqnarray}\label{eq:mse}
        \E\left[(\hat\tau^{\phantom{1}}_{EARL} - GATE)^2\right] = \frac 1 {N^2} \sum_{a\in\mathcal{A}}\sum_{a'\in\mathcal{A}}\E[(\hat\gamma_a-\gamma_a)(\hat\gamma_{a'}-\gamma_{a'})].
    \end{eqnarray}
    A given analysis unit $a$ has at most $d_{\mathcal A}d_{\mathcal R}$ analysis units with a shared connected randomisation unit in the experiment graph. Therefore, Lemma~\ref{lem:cov=0}(b) implies that the inner sum on the right-hand side of \eqref{eq:mse} has at most $d_{\mathcal A}d_{\mathcal R}$ non-zero terms for each $a$. At the same time, using the Cauchy--Schwarz inequality and Lemma~\ref{lem:kth-abs-moment-bound} ($k=2$) we can bound every such term as follows,
\begin{eqnarray*}
    \left|\E[(\hat\gamma_{a}-\gamma_{a})(\hat\gamma_{a'}-\gamma_{a'})]\right| &\le&
    \sqrt{\E[(\hat\gamma_{a}-\gamma_{a})^2]\,\E[(\hat\gamma_{a'}-\gamma_{a'})^2]} \nonumber\\
    &\le& \sqrt{\E[Y_a^2]\,\E[Y_{a'}^2]}\left(2K_w\right)^2
    \le 4L\,K_w^2
\end{eqnarray*}
for any $a,\,a'\in\mathcal{A}$.
It then follows that
\begin{eqnarray}\label{eq:O}
    \frac 1 {N^2} \sum_{a\in\mathcal{A}}\sum_{a'\in\mathcal{A}}\E[(\hat\gamma_a-\gamma_a)(\hat\gamma_{a'}-\gamma_{a'})] &\le& \frac{d_{\mathcal A}d_{\mathcal R}}{N}\, 4L\,K_w^2 \nonumber\\ &=& O(d_{\mathcal R}^{\phantom{1}}d_{\mathcal A}^3/N).
\end{eqnarray}
Together with \eqref{eq:mse} and the condition $d_{\mathcal R}^{\phantom{1}}d_{\mathcal A}^3=o(N)$, the established relation \eqref{eq:O} implies that
$$
\E\left[(\hat\tau^{\phantom{1}}_{EARL} - GATE)^2\right]\to 0\qquad \textup{as }N\to\infty,
$$
and the consistency of $\hat\tau^{\phantom{1}}_{EARL}$ for GATE follows.
\end{proof}

\subsection{Asymptotic Normality of the EARL Estimator}

\begin{proof}[Proof of Theorem~\ref{th:asymptotic-normality}]
    Let $F_N$ be the distribution function of the statistic
    $$
        t_N = {(\hat\tau^{\phantom{1}}_{EARL} - GATE)}/\sqrt{V_N},
    $$
    where $V_N$ is a shorthand for $\var(\hat\tau^{\phantom{1}}_{EARL})$,
    and let $W$ denote the Wasserstein distance.
    Similarly to \cite[Theorem~4.3]{harshaw2023erl}, we consider the variables $X_a = (\hat\gamma_a-\gamma_a)$. By Remark~\ref{rmk:unbiasedness}, we have that $\E[X_a] = 0$ for all $a\in\mathcal{A}$. Moreover, Lemma~\ref{lem:kth-abs-moment-bound} implies that the $X_a$, $a\in\mathcal{A}$, have finite fourth moments, $\E[X_a^4]\le (2K_w)^4L$, and $\E[|X_a|^3]\le (2K_w)^3\E[|Y_a|^3] \le (2K_w)^3L^{3/4}$ (using $(\E[|Y_a|^3])^{1/3}\le (\E[Y_a^4])^{1/4}$). It also holds for $\sigma^2 := \var(\sum_{a\in\mathcal{A}}X_a)$ that $\sigma^2 = N^2\var(\hat\tau^{\phantom{1}}_{EARL})=N^2V_N$. Finally,
    under Assumptions A\ref{a:mu} and A\ref{a:no-autocorrelation}$^*$ we have that
    \begin{equation*}
        X_a \ind \{X_{a'}\mid a'\in\mathcal A\colon \mathcal D(a')\cap \mathcal D(a)=\emptyset\}\qquad\forall\,a\in\mathcal A,
    \end{equation*}
    and therefore for every $a\in\mathcal A$ the variable $X_a$ has a dependency neighbourhood (in the sense of \cite[Definition~3.1]{Ro:11}) with at most $D:=d_{\mathcal R}d_{\mathcal A}$ elements. Applying \cite[Theorem~3.5]{Ro:11}, we obtain
    \begin{equation*}
        W(F_N,\,\mathcal{N}(0,\,1)) \le \frac{D^2}{\sigma^3}\sum_{a\in\mathcal A}\E[|X_a|^3] + \frac{\sqrt{26}D^{3/2}}{\sqrt{\pi}\,\sigma^2}\sqrt{\sum_{a\in\mathcal A}\E[
        X_a^4]}.
    \end{equation*}
    Now using the moment bounds above and Assumption~A\ref{a:non-degeneracy} (which gives $1/V_N \le N/\delta$ for all sufficiently large $N$), we further derive
    \begin{eqnarray*}
        W(F_N,\,\mathcal{N}(0,\,1)) &\le&
        \frac{D^2(2K_w)^3L^{3/4}}{N^2V_N^{3/2}}
        + \sqrt{\frac{26}{\pi}}\cdot\frac{D^{3/2}(2K_w)^2\sqrt{L}}{N^{3/2}V_N}\\
        &\le&\frac{8\,D^2K_w^3L^{3/4}}{\delta^{3/2}\sqrt N}
        + \sqrt{\frac{26}{\pi}}\cdot\frac{4\,D^{3/2}K_w^2\sqrt{L}}{\delta\sqrt N}
        \\
        &=& O\left(\frac{d^2_{\mathcal R}d^5_{\mathcal A}}{\sqrt{N}} + \frac{d_{\mathcal R}^{3/2}d_{\mathcal A}^{7/2}}{\sqrt{N}}\right)
        \;=\; O\left(\frac{d^2_{\mathcal R}d^5_{\mathcal A}}{\sqrt{N}}\right)\!.
    \end{eqnarray*}
    Therefore, if $d^{4}_{\mathcal{R}}d^{10}_{\mathcal{A}}=o(N)$, we have that
    \begin{equation*}
    W(F_N,\,\mathcal{N}(0,\,1)) \to 0\qquad \textup{as }N\to\infty,
    \end{equation*}
    and $\hat\tau^{\phantom{1}}_{EARL}$ is asymptotically normal.
\end{proof}

\subsection{A Concentration Lemma for Sparse Locally Dependent Sums}

The proofs of Theorems~\ref{th:asymptotic-variance} and~\ref{th:randomisation-inference} rest on the following elementary concentration bound.

\begin{lemma}\label{lem:sparse}
    Let $\mathcal Q$ be a finite index set and $\{\eta_\pi\}_{\pi\in\mathcal Q}$ a collection of square-integrable random variables with $\E[\eta_\pi^2]\le B$ for all $\pi\in\mathcal Q$. Suppose that for every $\pi\in\mathcal Q$,
    \begin{equation*}
        \#\{\rho\in\mathcal Q\colon \cov(\eta_\pi,\,\eta_\rho)\neq 0\} \le \kappa.
    \end{equation*}
    Then, for any $n\ge 1$,
    \begin{equation*}
        \frac 1{n^2}\sum_{\pi\in\mathcal Q}(\eta_\pi - \E[\eta_\pi]) = O_p\left(\frac{\sqrt{|\mathcal Q|\,\kappa\, B}}{n^2}\right).
    \end{equation*}
\end{lemma}
\begin{proof}
    By the Cauchy--Schwarz inequality, every non-zero covariance satisfies $|{\cov(\eta_\pi,\eta_\rho)}|\le \sqrt{\var(\eta_\pi)\var(\eta_\rho)}\le B$. Hence
    \begin{equation*}
        \var\left(\frac 1{n^2}\sum_{\pi\in\mathcal Q}\eta_\pi\right) = \frac 1{n^4}\sum_{\pi\in\mathcal Q}\sum_{\rho\in\mathcal Q}\cov(\eta_\pi,\,\eta_\rho) \le \frac{|\mathcal Q|\,\kappa\,B}{n^4},
    \end{equation*}
    and the claim follows from Chebyshev's inequality.
\end{proof}

\subsection{Unbiased Variance Estimation}\label{app:variance}

Throughout this subsection, Assumption~A\ref{a:4}$^*$ is in force; recall that it implies Assumptions~A\ref{a:mu} and~A\ref{a:4} and that, with the convention $B_1^a=B_2^a=0$ for isolated units, $\mu_a(x,\,y) = B_0^a+B_1^ax+B_2^axy$ and $\gamma_a = B_2^a$. Isolated units satisfy $\phi_a=G_a=F_a=0$, so $\hat\gamma_a = 0 = \gamma_a$ and their diagonal pairs lie outside $\mathcal P$; the identities below therefore concern units with $\mathcal D(a)\neq\emptyset$. We will use the following elementary moment identities, valid under Assumption~A\ref{a:1} for every non-isolated $a$ (all sums over $r,\,r'\in\mathcal D(a)$, cross terms vanish because $\E[\tilde Z_r]=0$ and the indicators are independent):
\begin{equation}\label{eq:phi-moments}
    \E[\phi_a] = 0,\qquad \E[\phi_a G_a] = 0,\qquad \E[\phi_a F_a] = 1,
\end{equation}
where the last identity follows from
$\E[(S_r/q)\tilde Z_r\,S_{r'}Z_{r'}] = \mathbf 1\{r=r'\}\E[S_r/q\cdot S_r]\E[\tilde Z_rZ_r] = \mathbf 1\{r=r'\}$, summed over $r'$ and divided by $|\mathcal D(a)|$. In particular, under Assumption~A\ref{a:4}$^*$,
\begin{equation}\label{eq:gamma-unbiased-astar}
    \E[\hat\gamma_a] = \E[\phi_a e_a(\mathbf S,\,\mathbf Z)] = B_1^a\E[\phi_aG_a] + B_2^a\E[\phi_aF_a] = B_2^a = \gamma_a.
\end{equation}
Analogously, for any two analysis units,
\begin{equation}\label{eq:A12}
    \E[\phi_{a_1}\phi_{a_2}] = \frac{|\mathcal D(a_1)\cap\mathcal D(a_2)|}{q\,p(1-p)},
\end{equation}
which vanishes for non-overlapping pairs.

\begin{proof}[Proof of Theorem~\ref{th:asymptotic-variance}]
    \emph{Step 1: unbiasedness.} Fix a pair $(a_1,\,a_2)\in\mathcal P$ with $a_1\neq a_2$ and write $S_{a_1,a_2} := \mathbf c_{a_1,a_2}^{\top}(\mathbf m_{a_1,a_2}-\E[\mathbf m_{a_1,a_2}])$ for the subtracted term in \eqref{def:R}. By construction, $\mathbf c_{a_1,a_2}$ solves the moment equations
    \begin{equation}\label{eq:moment-system}
        \E\left[S_{a_1,a_2}\, f_k(a_1)f_l(a_2)\right] = \mathbf 1\{k=l=2\},\qquad k,\,l\in\{0,\,1,\,2\},
    \end{equation}
    where $f_0(a)=1$, $f_1(a)=G_a$, $f_2(a)=F_a$. Indeed, for $(k,\,l)\neq(0,\,0)$ the product $f_k(a_1)f_l(a_2)$ is a component of $\mathbf m_{a_1,a_2}$, and since $S_{a_1,a_2}$ is centred,
    \begin{eqnarray*}
        \E\left[S_{a_1,a_2}\, f_k(a_1)f_l(a_2)\right] &=& \cov\left(S_{a_1,a_2},\,f_k(a_1)f_l(a_2)\right)\\
        &=& \left(\Sigma_{a_1,a_2}\mathbf c_{a_1,a_2}\right)_{f_kf_l},
    \end{eqnarray*}
    so the eight equations \eqref{eq:moment-system} with $(k,\,l)\neq(0,\,0)$ are coordinate-wise identical to the linear system $\Sigma_{a_1,a_2}\mathbf c = \mathbf u$ that defines $\mathbf c_{a_1,a_2}$ and hence hold by construction; for $(k,\,l)=(0,\,0)$ the equation reads $\E[S_{a_1,a_2}]=0$, which holds because $S_{a_1,a_2}$ is centred.

    Under Assumptions~A\ref{a:no-autocorrelation}$^*$ and A\ref{a:4}$^*$, conditioning on $(\mathbf S,\,\mathbf Z)$ and using Lemma~\ref{lem:cov=0}(a),
    \begin{eqnarray*}
        \E[Y_{a_1}Y_{a_2}S_{a_1,a_2}] &=& \E[e_{a_1}(\mathbf S,\,\mathbf Z)\,e_{a_2}(\mathbf S,\,\mathbf Z)\,S_{a_1,a_2}]\\
        &=& \sum_{k,\,l\in\{0,1,2\}}B^{a_1}_kB^{a_2}_l\,\E\left[S_{a_1,a_2}f_k(a_1)f_l(a_2)\right]\\
        &\stackrel{\eqref{eq:moment-system}}{=}& B_2^{a_1}B_2^{a_2} \;\stackrel{\eqref{eq:gamma-unbiased-astar}}{=}\; \E[\hat\gamma_{a_1}]\,\E[\hat\gamma_{a_2}].
    \end{eqnarray*}
    Since $Y_{a_1}Y_{a_2}\phi_{a_1}\phi_{a_2} = \hat\gamma_{a_1}\hat\gamma_{a_2}$ identically, we conclude
    \begin{equation}\label{eq:pair-unbiased}
        \E[Y_{a_1}Y_{a_2}R_{a_1,a_2}] = \E[\hat\gamma_{a_1}\hat\gamma_{a_2}] - \E[\hat\gamma_{a_1}]\E[\hat\gamma_{a_2}] = \cov(\hat\gamma_{a_1},\hat\gamma_{a_2}).
    \end{equation}
    For the diagonal terms, the same argument applies with the monomial set $(G_a,\,F_a,\,G_a^2,\,G_aF_a,\,F_a^2)$: the moment equations make $\E[e_a^2\,S_{a,a}] = (B_2^a)^2$, while $\E[\varepsilon_a^2 S_{a,a}]$ vanishes by Assumption~A\ref{a:no-autocorrelation}$^*$ and centredness of $S_{a,a}$, so, using $\E[Y_a^2\mid\mathbf S,\,\mathbf Z] = e_a^2 + \E[\varepsilon_a^2]$,
    \begin{equation*}
        \E[Y_a^2R_{a,a}] = \E[\hat\gamma_a^2] - \gamma_a^2 = \var(\hat\gamma_a).
    \end{equation*}
    Finally, for $(a_1,\,a_2)\notin\mathcal P$ we have $R_{a_1,a_2}=0$ and $\cov(\hat\gamma_{a_1},\hat\gamma_{a_2})=0$ by Lemma~\ref{lem:cov=0}(b). Summing over all pairs,
    \begin{equation*}
        \E[\hat V] = \frac 1{N^2}\sum_{a_1}\sum_{a_2}\cov(\hat\gamma_{a_1},\,\hat\gamma_{a_2}) = \var(\hat\tau^{\phantom{1}}_{EARL})
    \end{equation*}
    for every $N$, which is the first claim.

    \emph{Step 2: concentration.} We apply Lemma~\ref{lem:sparse} to the collection $\eta_\pi = Y_{a_1}Y_{a_2}R_{a_1,a_2}$, $\pi=(a_1,\,a_2)\in\mathcal P$, with $n=N$. First, the weights are uniformly bounded: $|\phi_{a_1}\phi_{a_2}|\le K_w^2$ by Lemma~\ref{lem:weight-bound}; all monomials in $\mathbf m_{a_1,a_2}$ take values in $[0,\,1]$, so $\|\mathbf m - \E[\mathbf m]\|_2\le \sqrt 8$, while $\|\mathbf c_{a_1,a_2}\|_2 \le \lambda_N^{-1}\|\mathbf u\|_2=\lambda_N^{-1}$ by Assumption~A\ref{a:lambda0}; hence
    \begin{equation}\label{eq:CR}
        |R_{a_1,a_2}| \le K_w^2 + \sqrt 8/\lambda_N =: C_R.
    \end{equation}
    Consequently, using the Cauchy--Schwarz inequality and $\E[Y_a^4]\le L$,
    \begin{equation*}
        \E[\eta_\pi^2] \le C_R^2\,\sqrt{\E[Y_{a_1}^4]\,\E[Y_{a_2}^4]} \le C_R^2L =: B.
    \end{equation*}
    Next, note that $|\mathcal P| \le Nd_{\mathcal A}d_{\mathcal R}$ (because $|\mathcal T(a)|\le d_{\mathcal A}d_{\mathcal R}$ for any $a\in\mathcal A$). Under Assumptions~A\ref{a:mu} and A\ref{a:no-autocorrelation}$^*$, each $\eta_{(a_1,a_2)}$ is a function of the indicators $\{(S_r,\,Z_r)\colon r\in\mathcal D(a_1)\cup\mathcal D(a_2)\}$ and the errors $(\varepsilon_{a_1},\,\varepsilon_{a_2})$; therefore $\eta_\pi$ and $\eta_\rho$ are independent (hence uncorrelated) unless the two pairs share an analysis unit or a connected randomisation unit. For a fixed $\pi=(a_1,\,a_2)$, any such $\rho=(a_3,\,a_4)$ must have $a_3$ or $a_4$ in $\mathcal T(a_1)\cup\mathcal T(a_2)\cup\{a_1,\,a_2\}$, a set of at most $2(d_{\mathcal A}d_{\mathcal R}+1)$ elements, and its other coordinate ranges over at most $d_{\mathcal A}d_{\mathcal R}$ units (to stay within $\mathcal P$). Hence
    \begin{equation*}
        \kappa \le 2\cdot 2(d_{\mathcal A}d_{\mathcal R}+1)d_{\mathcal A}d_{\mathcal R}\le 8\,d^2_{\mathcal A}d^2_{\mathcal R}.
    \end{equation*}
    Lemma~\ref{lem:sparse} now gives
    \begin{eqnarray*}
        \hat V - \var(\hat\tau^{\phantom{1}}_{EARL}) &=& O_p\left(\frac{\sqrt{Nd_{\mathcal A}d_{\mathcal R}\cdot 8d^2_{\mathcal A}d^2_{\mathcal R}\cdot C_R^2L}}{N^2}\right)\\ &=& O_p\left(\frac{(d_{\mathcal A}d_{\mathcal R})^{3/2}C_R}{N^{3/2}}\right).
    \end{eqnarray*}
    Dividing by $\var(\hat\tau_{EARL})\ge\delta/N$ (Assumption~A\ref{a:non-degeneracy}),
    \begin{equation*}
        \frac{\hat V}{\var(\hat\tau^{\phantom{1}}_{EARL})} - 1 = O_p\left(\frac{(d_{\mathcal A}d_{\mathcal R})^{3/2}\,C_R}{\sqrt N}\right) = o_p(1),
    \end{equation*}
    where, by \eqref{eq:CR}, the bound splits into a term of order $d^{7/2}_{\mathcal A}d^{3/2}_{\mathcal R}/\sqrt N$, which vanishes because $d^{7}_{\mathcal A}d^{3}_{\mathcal R}\le d^{10}_{\mathcal A}d^{4}_{\mathcal R} = o(N)$, and a term of order $(d_{\mathcal A}d_{\mathcal R})^{3/2}/(\lambda_N\sqrt N)$, which vanishes by the rate condition of Assumption~A\ref{a:lambda0}.

    \emph{Step 3: asymptotic validity.} Assumption~A\ref{a:4}$^*$ implies Assumptions~A\ref{a:mu} and A\ref{a:4}, so all conditions of Theorem~\ref{th:asymptotic-normality} hold and $(\hat\tau_{EARL}-GATE)/\sqrt{V_N}\stackrel{d}{\to}\mathcal N(0,\,1)$. Writing
    \begin{equation*}
        \frac{\hat\tau^{\phantom{1}}_{EARL} - GATE}{\sqrt{\hat V}} = \frac{\hat\tau^{\phantom{1}}_{EARL} - GATE}{\sqrt{V_N}}\cdot\sqrt{\frac{V_N}{\hat V}}
    \end{equation*}
    and applying Slutsky's theorem with $\hat V/V_N\stackrel{p}{\to}1$ completes the proof.
\end{proof}

\subsection{Randomisation Inference}\label{app:ri}

\begin{proof}[Proof of Theorem~\ref{th:randomisation-inference}]
    Under the strong null \eqref{def:strong-null} write $e_a$ for the constant value of $e_a(\cdot,\,\cdot)$, so that $Y_a = e_a + \varepsilon_a$ with, by Assumption~A\ref{a:no-autocorrelation}$^*$, the errors jointly independent of each other and of $(\mathbf S,\,\mathbf Z)$. We will repeatedly use the bound $\E[|Y_{a}Y_{b}Y_{c}Y_{d}|]\le L$, valid for any four (not necessarily distinct) analysis units by the generalised H\"older inequality and $\E[Y_a^4]\le L$; in particular, $\E[Y_{a_1}^2Y_{a_2}^2]\le L$. Denote $A_{a_1,a_2} := \E[\phi_{a_1}\phi_{a_2}]$, which by \eqref{eq:A12} vanishes unless $(a_1,\,a_2)\in\mathcal P$ and satisfies $|A_{a_1,a_2}|\le d_{\mathcal A}/(q\,p(1-p)) = K_w$ (cf.~\eqref{def:Kw}).

    \emph{Step 1: the target variance.} Under the null, $\E[\phi_aY_a] = e_a\E[\phi_a] = 0$, and for any $a_1,\,a_2$ (including $a_1=a_2$), using the independence of the errors from $(\mathbf S,\,\mathbf Z)$,
    \begin{eqnarray*}
        \E[\phi_{a_1}\phi_{a_2}Y_{a_1}Y_{a_2}] &=& \E[\phi_{a_1}\phi_{a_2}]\,\E[Y_{a_1}Y_{a_2}].
    \end{eqnarray*}
    Consequently,
    \begin{equation}\label{eq:var-null}
        \var(\hat\tau^{\phantom{1}}_{EARL}) = \frac 1{N^2}\sum_{a_1}\sum_{a_2}A_{a_1,a_2}\,\E[Y_{a_1}Y_{a_2}] =: \tilde V.
    \end{equation}

    \emph{Step 2: the infinite-$J$ limit of $\hat V_{RI}$.} Conditionally on the experimental data $\mathcal G$, by which we mean the observed outcomes $\{Y_a\}_{a\in\mathcal A}$ together with $(\mathbf S,\,\mathbf Z)$, the copies $W_j := \hat\tau_{EARL}(\tilde{\mathbf S}_j,\,\tilde{\mathbf Z}_j)$, $j=1,\ldots,J$, are independent and identically distributed with
    \begin{equation}\label{eq:vD}
        \E[W_j\mid \mathcal G] = 0,\qquad \var(W_j\mid\mathcal G) = \frac 1{N^2}\sum_{a_1}\sum_{a_2}A_{a_1,a_2}Y_{a_1}Y_{a_2} =: v.
    \end{equation}
    We claim $v/\tilde V\stackrel{p}{\to}1$. Indeed, $\E[v]=\tilde V$, and
    \begin{equation*}
        v - \tilde V = \frac 1{N^2}\sum_{(a_1,a_2)\in\mathcal P}A_{a_1,a_2}\left(Y_{a_1}Y_{a_2} - \E[Y_{a_1}Y_{a_2}]\right).
    \end{equation*}
    Under the null, $Y_{a_1}Y_{a_2}$ is a function of $(\varepsilon_{a_1},\,\varepsilon_{a_2})$ alone; by the joint independence of the errors, two terms $\eta_{(a_1,a_2)} = A_{a_1,a_2}Y_{a_1}Y_{a_2}$ and $\eta_{(a_3,a_4)}$ are independent unless $\{a_1,\,a_2\}\cap\{a_3,\,a_4\}\neq\emptyset$, which for a fixed pair leaves at most $\kappa = 4d_{\mathcal A}d_{\mathcal R}$ partners in $\mathcal P$: the shared unit can occupy either coordinate of the partner pair and equal $a_1$ or $a_2$, while the remaining coordinate must lie in the shared unit's $\mathcal T$-set for the partner to belong to $\mathcal P$. With $\E[\eta_\pi^2]\le K_w^2\,\E[Y_{a_1}^2Y_{a_2}^2]\le K_w^2L$ and $|\mathcal P|\le Nd_{\mathcal A}d_{\mathcal R}$ (as in the proof of Theorem~\ref{th:asymptotic-variance}), Lemma~\ref{lem:sparse} yields
    \begin{equation*}
        v - \tilde V = O_p\left(\frac{\sqrt{4N(d_{\mathcal A}d_{\mathcal R})^2K_w^2L}}{N^2}\right) = O_p\left(\frac{d_{\mathcal A}d_{\mathcal R}K_w\sqrt L}{N^{3/2}}\right),
    \end{equation*}
    and dividing by $\tilde V = \var(\hat\tau_{EARL})\ge \delta/N$ (Assumption~A\ref{a:non-degeneracy}) gives
    \begin{equation}\label{eq:v-ratio}
        \frac{v}{\tilde V} - 1 = O_p\left(\frac{d^2_{\mathcal A}d_{\mathcal R}}{\sqrt N}\right)\cdot\frac{\sqrt L}{\delta\, q\,p(1-p)} = o_p(1),
    \end{equation}
    since $d^4_{\mathcal A}d^2_{\mathcal R}\le d^{10}_{\mathcal A}d^4_{\mathcal R} = o(N)$.

    \emph{Step 3: finite $J$.} We show $\hat V_{RI}/v\stackrel{p}{\to}1$. Write $\hat V_{RI} = T - \bar W^2$ with $T := J^{-1}\sum_j W_j^2$ and $\bar W := J^{-1}\sum_jW_j$. Conditionally on $\mathcal G$, the variables $W_j^2$ are i.i.d.\ with mean $v$, so $\E[T\mid\mathcal G] = v$ and
    \begin{equation}\label{eq:cond-mc}
        \var\left(T\mid \mathcal G\right) \le \frac{\E[W_1^4\mid\mathcal G]}{J},\qquad \E[\bar W^2\mid \mathcal G] = \frac vJ.
    \end{equation}
    It remains to bound the fourth moment. Let $V_a := \phi_a(\tilde{\mathbf S}_1,\,\tilde{\mathbf Z}_1)Y_a$, so that $W_1 = N^{-1}\sum_aV_a$ and $\E[V_a\mid\mathcal G]=0$. Under Assumption~A\ref{a:1} and the re-randomisation scheme \eqref{def:S-Z-samples}, $\phi_a(\tilde{\mathbf S}_1,\,\tilde{\mathbf Z}_1)$ is independent of $\mathcal G$ and of $\{\phi_{a'}(\tilde{\mathbf S}_1,\,\tilde{\mathbf Z}_1)\colon a'\notin\mathcal T(a)\}$ jointly, and $|\phi_a(\tilde{\mathbf S}_1,\,\tilde{\mathbf Z}_1)|\le K_w$ by Lemma~\ref{lem:weight-bound}. Expand $\E[(\sum_aV_a)^4\mid\mathcal G]$ into the sum of $\E[V_aV_bV_cV_d\mid\mathcal G]$ over all quadruples of indices; any term in which one index lies outside the $\mathcal T$-neighbourhoods of the other three vanishes (its weight factor is independent of the rest and centred, while the outcomes are $\mathcal G$-measurable), and every remaining term is bounded by $K_w^4\,|Y_aY_bY_cY_d|$. The number of remaining terms is at most $3N^2\Delta^2 + O(N\Delta^3) = O(N^2\Delta^2)$ with $\Delta := d_{\mathcal A}d_{\mathcal R}$ (choose an index and a neighbour of it, then another index and a neighbour of it, for terms splitting into two dependent pairs; the remaining connected configurations number at most $cN\Delta^3$ for some absolute constant $c$ independent of $N$); the final step uses $\Delta\le N$, which holds for all sufficiently large $N$ because of the sparsity condition. Taking expectations and using $\E[|Y_aY_bY_cY_d|]\le L$,
    \begin{equation*}
        \E[W_1^4] = \E\bigl[\E[W_1^4\mid \mathcal G]\bigr] = O\left(\frac{N^2\Delta^2 K_w^4L}{N^4}\right) = O\left(\frac{\Delta^2K_w^4L}{N^2}\right).
    \end{equation*}
    Since $\E[T-v\mid\mathcal G]=0$, the tower property and \eqref{eq:cond-mc} give $\E[(T-v)^2] = \E[\var(T\mid\mathcal G)]\le \E[W_1^4]/J$ and $\E[\bar W^2] = \E[v]/J = \tilde V/J$. Applying Markov's inequality to $(T-v)^2$ and to $\bar W^2$ then yields
    \begin{equation*}
        T - v = O_p\left(\frac{\Delta K_w^2\sqrt L}{N\sqrt J}\right),\qquad
        \bar W^2 = O_p\left(\frac{\tilde V}{J}\right).
    \end{equation*}
    Dividing by $v = \tilde V(1+o_p(1)) \ge (\delta/N)(1+o_p(1))$, the second term is $O_p(1/J) = o_p(1)$ and the first is
    \begin{equation*}
        O_p\left(\frac{\Delta K_w^2\sqrt L }{\delta\sqrt J}\right) = O_p\left(\frac{d^3_{\mathcal A}d_{\mathcal R}}{\sqrt{J}}\right)\cdot\frac{\sqrt L}{\delta\,(q\,p(1-p))^2} = o_p(1),
    \end{equation*}
    by the assumption $d^6_{\mathcal A}d^2_{\mathcal R} = o(J_N)$. Combining with \eqref{eq:v-ratio},
    \begin{equation*}
        \frac{\hat V_{RI}}{\var(\hat\tau^{\phantom{1}}_{EARL})} = \frac{\hat V_{RI}}{v}\cdot\frac{v}{\tilde V}\stackrel{p}{\to}1.
    \end{equation*}

    \emph{Step 4: size of the test.} Under the null, $H_0^*$ implies Assumption~A\ref{a:4} (with all contributions constant) and $GATE=0$, so Theorem~\ref{th:asymptotic-normality} applies (its moment condition $\E[Y_a^4]\le L$ holds by assumption), making $\hat\tau_{EARL}$ asymptotically standard normal after standardisation by the square root of its variance. Writing
    \begin{equation*}
        \frac{\hat\tau^{\phantom{1}}_{EARL}}{\sqrt{\hat V_{RI}}} = \frac{\hat\tau^{\phantom{1}}_{EARL}}{\sqrt{\var\left(\hat\tau^{\phantom{1}}_{EARL}\right)}} \cdot \sqrt{\frac{\var\left(\hat\tau^{\phantom{1}}_{EARL}\right)}{\hat V_{RI}}}
    \end{equation*}
    and applying Slutsky's theorem gives $\hat\tau_{EARL}/\sqrt{\hat V_{RI}}\stackrel{d}{\to}\mathcal N(0,\,1)$, hence
    \begin{equation*}
        \Pr\left\{\left|\hat\tau^{\phantom{1}}_{EARL}\right| > z_{1-\alpha/2}\sqrt{\hat V_{RI}}\right\}\to \alpha,
    \end{equation*}
    completing the proof.
\end{proof}

\end{document}